\documentclass[letterpaper,12pt]{article}
\usepackage{amsmath, mathtools}
\usepackage{amssymb}
\usepackage{amsthm}
\usepackage{graphicx, placeins,  changes, hyperref}
\usepackage{verbatim}
\usepackage{color}
\usepackage[margin=1in]{geometry}
\usepackage{tikz, tikz-cd}
\usetikzlibrary{fit, positioning, arrows, shapes}
\usepackage[affil-it]{authblk}

\usepackage{algpseudocode}
\usepackage{algorithm}
\newtheorem{theorem}{Theorem}
\newtheorem{definition}{Definition}

\newtheorem{lemma}{Lemma}
\newtheorem{remark}{Remark}
\newtheorem{prop}{Proposition}

\def\P{{\mathbb P}}     
\def\E{{\mathbb E}}     
\def\RR{\mathbb{R}} 

\def\ZZ{\mathbb{Z}}

\definecolor{Red}{rgb}{1,0,0}
\definecolor{Blue}{rgb}{0,0,1}

\title{Joint identifiability of ancestral sequence, phylogeny and mutation rates under the TKF91 model}

\date{\today}

\author[1]{Alex Xue}
\author[2,*]{Brandon Legried}
\author[3,4,*]{Wai-Tong Louis Fan}

\affil[1]{\small Department of Mathematics, University of California, Los Angeles, CA, USA.}
\affil[2]{School of Mathematics, Georgia Institute of Technology, Atlanta, GA, USA.}
\affil[3]{Department of Mathematics, Indiana University, Bloomington, IN, USA.}
\affil[4]{Department of Organismic and Evolutionary Biology, Harvard University, Cambridge, MA, USA.}

\affil[*]{Corresponding authors: L. Fan (wlfan@fas.harvard.edu), B. Legried (blegried3@gatech.edu)}

\begin{document}

\maketitle

\begin{abstract}
We consider the problem of identifying jointly the ancestral sequence, the phylogeny and the parameters in models of DNA sequence evolution with insertion and deletion (indel). 
Under the classical TKF91 model of sequence evolution, we obtained explicit formulas for the root sequence, the pairwise distances of leaf sequences, as well as the scaled rates of indel and substitution in terms of the distribution of the leaf sequences of an arbitrary  phylogeny. These explicit formulas not only strengthen existing invertibility results and work for phylogeny that are not necessarily ultrametric, but also
lead to new estimators 
that require less assumptions. Our simulation study demonstrates that these estimators are statistically consistent as the number of independent samples increases. 
\end{abstract}

\section{Introduction}

In molecular phylogenetics, the two main problem classes are ancestral state reconstruction and phylogenetic tree reconstruction. In these problems, a rooted tree is assumed to describe the dependency structure of species evolution. The root represents an initial ancestral sequence (usually DNA, RNA or amino acid sequences) that evolves according to a Markov process, with the tree serving as a generalized timeline. Ancestral state reconstruction involves inferring the parameters of the sequence model based on the distribution of sequences observed at the tree's leaves, or perhaps a marginalization of this distribution, such as sequence lengths or $k$-mer counts. These parameters typically include mutation rates—such as substitution, insertion, and deletion rates—as well as the root sequence, if fixed. In contrast, tree reconstruction focuses on inferring the tree's branching pattern (topology) and the lengths (metric) of the edges, usually under the assumption that mutation rates are known up to time-scaling.

The most widely studied sequence evolution models, such as the Cavender-Farris-Neyman, Jukes-Cantor, and General Time-Reversible models, primarily focus on substitution mutations. They are typically time-reversible and often assume that the root sequence follows the stationary distribution. In substitution-only models, sites are assumed to evolve independently and identically distributed (i.i.d.). When sequence lengths are fixed across species, much is already known about the fundamental limits of inference under this i.i.d. assumption. With perfectly aligned columns, a simple continuous-time Markov chain is typically assigned to each site, and the evolution of the sites is treated as independent. 

As expected, the problem becomes much more challenging when insertions and deletions (indels) are introduced. Sites can no longer be treated as independent, as their placement in the sequence provides information about the positioning of other sites. The classical Thorne-Kishino-Felsenstein (TKF91) model of sequence evolution \cite{Thorne1991} extends substitution models while preserving time-reversibility. Given the growing importance of modeling indels in both biology and bioinformatics \cite{redelings2007}, \cite{warnow2012}, \cite{Warnow:17}, numerous generalizations and simplifications of the TKF91 model have been motivated and introduced (see, for example, \cite{thorne1992inching}, \cite{miklos2004long}, \cite{mitrophanov2007convergence}, \cite{bouchard-cote2013PIP}). A key mathematical difficulty in establishing identifiability arises from the need to estimate multiple unknown rate parameters jointly, particularly when both insertion and deletion processes are present. In the context of estimating mutation rates from arbitrary birth and death histories, there have been both positive and negative results concerning identifiability, even with ideal sequence data \cite{louca2020}, \cite{legried2022rateIdent}, \cite{legried2023IdentInfer}.

Another important question, besides rate estimation and ancestral state estimation, is the estimation of an unknown phylogeny, or phylogenetic tree. For substitution-only models, as the number of sequences $N$ goes to infinity, it was shown in 1999 that only $\ln N$
sites are required to reconstruct the phylogenetic tree with high probability (see \cite{erdos1999LogsI}, \cite{erdos1999LogsII}). In 2006, \cite{THATTE200658} demonstrated that for the TKF91 sequence model for indels, the phylogenetic tree with edge lengths is identifiable, given the mutation parameters and only the ability to observe sequence lengths at the leaves of the tree. The existence of a consistent estimator implies identifiability of parameters, but the reverse is not necessarily true. Identifiability only means that the parameter settings map injectively into their likelihood formulas. Additionally, the existence of a consistent estimator does not guarantee that the estimator can reconstruct a parameter with a practical amount of data. Indeed, obtaining consistent estimators from sequence lengths or $k$-mers
 remains challenging even for the simple cases of the TKF91 model and the Cavender-Farris-Neyman model (see \cite{fan2018necessary}, \cite{fanlegriedroch2020}, \cite{fan2022impossibilityKmers}).

There has been better success in tree and ancestral reconstruction when it is possible to characterize how indels affect the multiple sequence alignment \cite{DaskalakisRoch:13}, \cite{legried2023pairwise}. By modeling displacement, indel models like TKF91 provide a principled approach to jointly modeling the phylogenetic tree and a tertiary object known as the multiple sequence alignment. The stochastic process explains how "gaps" in the alignment arise, and understanding these gaps is crucial for determining which columns are most informative for phylogenetic reconstruction. However, constructing a multiple sequence alignment is difficult, prompting some to explore alignment-free estimation methods. In \cite{daskalakis2010alignment} and \cite{DaskalakisRoch:13}, it was observed that some correlative information about sites could be used for tree estimation. Another approach involves using  $k$-mer vectors \cite{allman2017statistically}, although \cite{fan2022impossibilityKmers} showed that statistical consistency fails for the CFN model. It is conjectured that this impossibility result for  $k$-mers in the fixed-length CFN model could extend to  $k$-mers observed under the TKF91 model.

In biological applications, joint identifiability and joint estimation are important because they offer a principled way to identify or estimate some parameters without assuming others are known. Estimating substitution probabilities along phylogenetic trees from sequence data can be challenging, potentially impacting reconstruction accuracy. From present-day sequence data, ancestral sequence reconstruction is typically required, where the past states of the model are inferred, followed by determining the rate at which sequences evolve. Computational challenges persist when indels are considered, as constructing a multiple sequence alignment becomes necessary \cite{metzler2003}, \cite{loytynoja2008}. Some pipelines exist in the presence of indels. For example, in \cite{ratan2015identification}, the indelMINER algorithm set was developed to detect the presence and absence of indels from whole-genome resequencing datasets using paired-end reads. However, this paper does not address long indel models such as \cite{miklos2004long}, which remains a future direction for this work.

In this paper, we investigate the identifiability of ancestral sequences, mutation rates, and the underlying phylogenetic tree, with a focus on the parameters governing insertion and deletion. We tackle two significant challenges—models with indels and joint identifiability—simultaneously using the TKF91 model. We focus on the TKF91 model because, despite its relative simplicity among indel models, it remains poorly understood, and we hope that our analysis will provide insights into more realistic indel models. To the best of our knowledge, there are no existing results that establish joint identifiability of rate parameters in indel models, nor tree identifiability with an arbitrary initial sequence. This paper not only provides affirmative answers to these questions of joint identifiability but also introduces explicit estimators that appear to be consistent, as demonstrated by our simulation studies. 

Our explicit formulas strengthen existing invertibility results from \cite{THATTE200658} and \cite{fan2020statistically}, applicable to arbitrary phylogenies (not necessarily ultrametric), and lead to new estimators with fewer assumptions. Our proof proceeds by (i) identifying suitable injective maps from large parameter spaces, such as those of potential ancestral sequences and model parameters, to the space of probability distributions for the leaf sequences, and (ii) deriving explicit inverses for these maps.
A key component of the proofs is the analysis of the 1-mer process. Our likelihood-based approach does not required stationary assumption as was assumed in existing works such as \cite{felsenstein1981,hasegawa1985, tavare1986,THATTE200658}.

\bigskip

{\bf Organization. } The rest of the paper is organized as follows. The indel process is defined in Section~\ref{sec:defs}.  Our main result, together with a proof sketch, is stated in Section~\ref{sec:results}. Details of the proof are provided in Section~\ref{Proofs}.
In Section~\ref{S:Simulation} we provide simulations that further confirm our results.

\section{Basic definitions}\label{sec:defs}

In this section, we describe the model of sequence evolution with insertion and deletion considered throughout this paper. 

A common biological assumption is that genetic material  of each species  $u$ can be represented by a sequence $\vec{x}^u=(x^u_1,\cdots, x^u_{M^u})$ of length $M^u$ over a finite alphabet. We work here with the binary alphabet $\{0,1\}$ for simplicity, following \cite{daskalakis2010alignment,DaskalakisRoch:13}. We refer to the positions of a digit in a sequence as a  site.

\begin{definition}[Binary indel process]\label{Def_BiINDEL}
	The  {\bf binary indel process} is a  continuous-time Markov process $\mathcal{I}=(\mathcal{I}_t)_{t\geq 0}$ on the space  $\bigcup_{M\geq 1} \{0,1\}^M$
of binary digit sequences together with the empty set $\emptyset$. That is, the state space is
	\begin{equation}\label{S}
		\mathcal{S} :=  \emptyset \cup \bigcup_{M\geq 1} \{0,1\}^M,
	\end{equation}
	The parameters of this model are $(\lambda,\mu,\nu) \in (0,\infty)^3$ with $\lambda\neq \mu$, and $(\pi_0,\,\pi_1)\in [0,1]^2$ with $\pi_0 + \pi_1 = 1$. The  Markovian dynamic is described as follows: if the current state is the sequence $\vec{x}$, then the following events occur independently:
	\begin{itemize}
		\item (Substitution)$\;$ Each digit  is substituted independently at rate $\nu>0$. When a substitution occurs, the corresponding digit is replaced by 0 and 1 with probabilities $\pi_0$ and $\pi_1$ respectively. 

		\item (Deletion)$\;$ Each digit  is removed independently at rate $\mu>0$.
		
		\item (Insertion) $\;$ Each digit gives birth to a new digit independently at rate $\lambda>0$. When a birth occurs, a digit is  added immediately to the right of its parent site. The newborn site has digit 0 and 1 with probabilities $\pi_0$ and $\pi_1$ respectively. 
	\end{itemize}
\end{definition}

\begin{remark}\rm
This model is a  simplified version of the classical TKF91 model~\cite{Thorne1991} of DNA sequence evolution, where we consider binary alphabet $\{0,1\}$ and we omit the immortal link. Such simplifications were first adopted 
\cite{daskalakis2010alignment,DaskalakisRoch:13}. 
We believe our method and results can be extended to richer alphabets, such as the four nucleotides in DNA/RNA sequences.

In Definition \ref{Def_BiINDEL} and throughout this paper, we assume $\lambda\neq\mu$ to avoid the special and biologically unrealistic case where 
$\lambda=\mu$. For instance, in RNA virus evolution, the insertion rate is more than 100 times smaller than the deletion rate \cite{aguilar2023high}. Similarly, in the early evolution of vertebrates, the duplication and loss of genomes showed a death rate approximately five times higher than the birth rate, as noted in \cite{cotton2005}. We believe Theorem \ref{T:main} still holds for this special case; see the Appendix. However, we leave the details of this special case to the interested reader.
\end{remark}

\medskip

Our analysis of the indel process makes heavy use of the underlying sequence-length  and the 1-mer count.

\begin{definition}
The \textbf{length} of a sequence $\vec{x}\in \mathcal{S}$ is defined as the number of digits in $\vec{x}$
and is denoted by $|\vec{x}|$. Therefore, if $\vec{x} = (x_1,...,x_M)$, then $|\vec{x}|=M$. The \textbf{1-mer count}  of a sequence $\vec{x}\in \mathcal{S}$ is defined as
the vector in $\mathbb{Z}_+^2$ whose first and second entries are respectively the numbers of 1's and 0's in the sequence. 
\end{definition}

\medskip

Let $T$ be a  binary metric tree with root  $\rho$ and
leaves set $\partial T$. Every edge may be viewed as a generalized timeline.   We are interested in the binary indel process starting at the root $\rho$ and running down the tree $T$,  described as a process $\{\mathcal{I}_u\}_{u\in T}$ in the next paragraph.   

The root vertex $\rho$ is assigned a state $\mathcal{I}_{\rho} \in \mathcal{S}$.  This state is then evolved down the tree according to the following recursive process.  Moving away from the root, along each edge $e = (u,v) \in E$, conditionally on the state $\mathcal{I}_{u}$, we run the indel process for a time $t_{u,v}$ which is the distance between $u$ and $v$ on the metric tree.  The process terminates upon reaching the leaf set $\partial T$.  Denote by $\mathcal{I}_{t}$ the resulting state at $t \in e$.  
In particular, the set of leaf states is $\mathcal{I}_{\partial T} = \{\mathcal{I}_{u}\}_{u \in \partial T}$.

Observe that the sequence length and the 1-mer counts of the binary indel process (in Definition \ref{Def_BiINDEL}) on a tree $T$  are themselves Markov processes on $T$. This is not true for the process of $k$-mer count for $k\geq 2$ 
\cite[Definition 1]{fan2022impossibilityKmers}. For example, the process of $2$-mer count is non-Markovian.

\bigskip

\noindent
{\bf Notation. }
Throughout the paper, we let $\mathbb{Z}_{+}=\{0,1,2,3,...\}$. 
We let $\mathbb{P}_{\vec{x}}$ be the probability measure when the root state is $\vec{x}$.  If the root state is chosen according to a distribution $\nu$, then we denote the probability measure by $\mathbb{P}_{\nu}$. We also denote  by $\mathbb{P}_{M}$ the conditional probability measure given the event that the root state has length $M$.
Under $\pi_{\rm sim}^{\otimes M}$, the sequence length is $M$ and all digits are independent with distribution $\pi_{sim}:=(\pi_0,\pi_1)$. Note that $\P_{\pi_{sim}^{\otimes M}} \neq  \P_{M}$ since $\P_M$ does not assume the digits of the initial sequence are i.i.d. For any random variable $R$, we let $$\mathcal{L}_{\vec{x}}(R):= \P_{\vec{x}}(R\in \cdot)$$
be the distribution of $R$ under $\mathbb{P}_{\vec{x}}$. 

Note that the \textbf{length process} $(|\mathcal{I}_{t}|)_{t\geq 0}$ is itself a linear birth-death process, with birth-rate $\lambda$ and death rate $\mu$. Therefore, $\left\{(|\mathcal{I}_{t}|)_{t\geq 0} ,\,\{\P_M\}_{M\in\ZZ_+}\right\}$ is  a Markov family. 
Clearly, the expected length at time $t$ is
\begin{equation}\label{E:expectedLength}
    \E_{M}[|\mathcal{I}_t|] = M e^{-(\mu-\lambda)t}
\end{equation}
for $t\in\RR_+,\,M\in\ZZ_+$ and $(\mu,\lambda)\in(0,\infty)^2$.

\section{Main result}
\label{sec:results}

Our main result is the following joint identifiability of the ancestral state $\vec{x}$, the phylogeny $T$ and the parameters of the indel process in Definition \ref{Def_BiINDEL}.  

\begin{theorem}\label{T:main}
The topology and all the edge lengths  (up to the multiplicative factor $\mu$)  of the phylogeny $T$ can be identified using only the laws  $\left\{\mathcal{L}_{\vec{x}}\left(|\mathcal{I}_{u}|,\,|\mathcal{I}_{v}| \right)\right\}_{(u,v)\in \partial T\times \partial T}$ of pairwise sequence lengths. For each leaf $u\in\partial T$,
the parameter $(\mu t_u,\lambda t_u)$ is identifiable using only the law  $\mathcal{L}_{\vec{x}}\left(\mathcal{I}_{u} \right)$ of the sequence at $u$. Suppose, furthermore, $\pi_{0}$ is known. Then the initial sequence $\vec{x}\in \bigcup_{M\geq 1} \{0,1\}^M$ and the parameter $\nu t_u$ 
are identifiable using the law  $\mathcal{L}_{\vec{x}}\left(\mathcal{I}_{u} \right)$ for any $u\in\partial T$. 
\end{theorem}

Furthermore, we obtain a constructive proof of this theorem which not only pinpoints what information is enough to identify the unknowns, but also gives explicit formulas for the identification. These formulas  offer intuition for constructing consistent estimators. In the simulation section, Section \ref{S:Simulation}, we demonstrate that these estimators are consistent as the number of independent samples gets large.

The requirement ``$\pi_{0}=1-\pi_1$ is known" will be needed to identify the scaled substitution rate $\nu t_u$ and will not be needed to identify the phylogeny.
Although it might be more satisfying to perform joint estimation when $(\pi_0,\pi_1)$ is unknown, it is a conventional 
assumption to treat this vector as known \cite{felsenstein1981,erdos1999LogsI,erdos1999LogsII}. We also offer a simple explanation to why we assumed that $\pi_0$ is known in Remark \ref{Rk:Impossibility_pi0}.

\begin{proof}[Proof of Theorem \ref{T:main}]
The desired result follows from Theorems
\ref{T:Initial_M}, \ref{T:Initial_1mer}, \ref{T:Initial_seq} and \ref{thm:phylogeny}.
We first identify the (scaled) insertion rate and deletion rate in Theorem \ref{T:Initial_M}, then the substitution rate in Theorem \ref{T:Initial_1mer}, and finally the ancestral sequence in Theorem \ref{T:Initial_seq}. In Theorem \ref{thm:phylogeny}, we identify the phylogeny using Theorem \ref{T:Initial_M}.
\end{proof}

\bigskip

\noindent
{\bf Relation to existing results. }
Recall that the indel process $\mathcal{I}=(\mathcal{I}_t)_{t\geq 0}$ is parameterized by  $\theta=(\pi_0,\,\nu,\,\lambda,\,\mu)\in \Theta$, where $\Theta=(0,1)\times (0,\infty)^3$.
The first question one may seek to answer is the following:
Is the map
$(t,\,\vec{x},\,\theta) \mapsto \mathcal{L}^{(\theta)}_{\vec{x}}(\mathcal{I}_t)=\mathcal{L}_{\vec{x}}(\mathcal{I}_t)$  injective on the domain $(0,\infty) \times \mathcal{S} \times \Theta$? 
The answer is trivially negative because we can always scale time by a constant factor without changing the distribution of the leaf states. Hence we should instead consider the map
$(\vec{x},\,t\theta) \mapsto \mathcal{L}^{(\theta)}_{\vec{x}}(\mathcal{I}_t)$ with domain $\mathcal{S} \times \Theta$. 

To date, the following related partial results are known.
 If we assume the initial sequence follows the stationary distribution $\Pi$ of the TKF91 model, then $\lambda/\mu$ can be identified from the expected length at equilibrium $\bar{L}=\E_{\Pi}[|\mathcal{I}_t|]$. Namely,  
 $$\lambda/\mu=\frac{\bar{L}}{1+\bar{L}}$$ 
 according to \cite[eqn.(33)]{THATTE200658}.  If $\theta$ is known,
    then the length $|\vec{x}|$ can be explicitly identified from $\E_{\vec{x}}[|\mathcal{I}_t|]$ as 
    \[
    |\vec{x}|= \frac{1}{\beta_t}\left(\E_{\vec{x}}[|\mathcal{I}_t|]-\frac{\gamma(1-\beta_t)}{1 -\gamma}\right),
    \]
    where $
    \beta_t:= e^{-(\mu-\lambda)t}$ and $\gamma:=\frac{\lambda}{\mu}$; 
    see \cite[Lemma 2(i)]{fan2020statistically}.
 If $\theta$ and $M=|\vec{x}|$ are known,
    then $\vec{x}$ can be explicitly identified from  the probability that
    the first nucleotide of $\mathcal{I}_t$ is $\sigma \in \{A, T, C,G\}$, when the initial state is $\vec x$; see \cite[Lemma 2(ii)]{fan2020statistically}.
 In trace reconstruction (see \cite{holden2020lower}, \cite{holden2020subpolynomial}), an unknown binary sequence $\vec{x}\in \{0,1\}^M$ of length $M$ is sent
through a noisy channel where insertion and deletion occur independently.  Consistency and sample complexity results hold for arbitrary sequences $\vec{x}$ that are sufficiently long.  By taking independent traces subject to deletions, insertions, and substitutions, correct reconstruction of $\vec{x}$ with high probability is possible for a sufficiently large $M = |\vec{x}|$. 

Despite these works, it was unclear whether evolutionary parameters and properties of the initial sequence could be jointly estimated from contemporary sequence information. Here, we address this more challenging question of joint identifiability when multiple parameters are unknown.  
In Theorem \ref{T:main}, we give an  affirmative answer assuming that we know the value of $\pi_0$ and that $\lambda\neq \mu$. More precisely, Theorem \ref{T:main} implies that
the map $\big(\vec{x},\,\nu t,\lambda t,\mu t\big) \mapsto \mathcal{L}^{(\theta)}_{\vec{x}}(\mathcal{I}_t)$   is injective on the domain $\bigcup_{M\geq 1} \{0,1\}^M \times (0,\infty) \times \left((0,\infty)^2\setminus \Delta\right)$. Below we give a proof of Theorem \ref{T:main} that also offers explicit formulas for statistical estimation.

We also consider existing results on phylogeny estimation. The method developed in this paper involves an arbitrary unknown sequence, allowing us to provide a new constructive proof for phylogeny estimation, given evolutionary rates, initial sequence length, and covariance of observed sequence lengths. The result in \cite{fanlegriedroch2020} shows that tree estimation is impossible from sequence lengths originating from a stationary root distribution, despite \cite{THATTE200658} demonstrating tree identifiability in this context. The impossibility result for $k$-mers  in the CFN model from \cite{fan2022impossibilityKmers} has not yet been extended to the TKF91 or other indel models, but the additional information provided by covariances suggests potential estimators for trees using $k$-mers  as well.

\subsection{Constructive proof for indel rates and sequence length}

First, note that although the mapping $\Phi: \,\ZZ_+\to \RR_+$ defined by $ \Phi(M)= \E_{M}[|\mathcal{I}_{t}|]$  is injective from \eqref{E:expectedLength}, its inverse $\Phi^{-1}$ depends on  $(\mu-\lambda)t$ which we \textit{do not know}.
Hence, the initial sequence length is not explicitly identifiable from the expected sequence length \textit{without} (jointly) identifying some parameters of the indel model. 

This is the challenge that we overcome in Theorem \ref{T:Initial_M} below, where we explicitly identify the unknown  triple $(|\vec{x}|,\,\mu t,\,\lambda t)$ from the first three moments of the length $|\mathcal{I}_{t}|$, where $|\vec{x}|$ is the  sequence length of the  initial/ancestral sequence $\vec{x}$, and $\mu$ and $\lambda$ are respectively the deletion rate and the insertion rate. Recall that we assume $(\lambda,\mu) \in (0,\infty)^2\setminus \Delta$ where $\Delta$ is the diagonal.  Removing the diagonal enforces the assumption that $\lambda \ne \mu$.

\begin{theorem}[Constructive proof for indel rates and sequence length]\label{T:Initial_M}
For any $t \geq 0$, the  mapping 
$\Phi_{t}^{(1)}:\,\ZZ_+\times \left( (0,\infty)^2\setminus \Delta\right) \to \RR_+^3$  defined by 		\begin{equation*}
		\Phi_{t}^{(1)}(M, \mu t, \lambda t)= \left(\E_{M}[|\mathcal{I}_{t}|],\,\E_{M}[|\mathcal{I}_{t}|^2],\,\E_{M}[|\mathcal{I}_{t}|^3] \right)
		\end{equation*}
is one-to-one, and its inverse $\left(\Phi_{t}^{(1)}\right)^{-1}$ does not depend on any parameter of the indel process, the time $t$, nor the initial sequence $\vec{x}$.
\end{theorem}

The mapping $\Phi^{(1)}_t$ in Theorem \ref{T:Initial_M}  has  an explicit expression for its inverse which can be described as follows. When  $\lambda\neq \mu$, set
\begin{equation}\label{ExplictInverses_length}
M=\frac{ \E_{M}[|\mathcal{I}_t|]}{\beta},\quad \mu t =\frac{-\ln \beta}{1-\gamma}, \quad \lambda t =\frac{-\gamma \ln \beta}{1-\gamma},
\end{equation}
where
\begin{align}\label{betagamma}
\beta = \frac{\gamma(2 + C_2') - C_2'}{1 + \gamma}, \qquad 
\gamma = \frac{\sqrt{-(C_2'+1)^2 \left(3 C_2'^2-2 C_3'\right)}-C_2'^2+C_2'+C_3'}{2 C_2'^2+2 C_2'-C_3'+2}
\end{align}
where  $C_2' = C_2/C_1$,  $C_3' = C_3/ C_1$ and
\begin{equation}\label{Def:Ck}
C_k:=\frac{\partial^{(k)}\ln G_M(z,t)}{\partial^{(k)} z}\Big|_{z=1}.
\end{equation}

\begin{remark}\label{Rk:Ck}\rm
Observe that
\begin{equation}\label{partialG}
G_M^{(k)}(1,t):=\frac{\partial^{(k)} G_M(z,t)}{\partial^{(k)} z}\Big|_{z=1}=\mathbb{E}_M\left[\frac{|\mathcal{I}_t|\,!}{(|\mathcal{I}_t|-k)!}1_{|\mathcal{I}_{t}| \geq k}\right],
\end{equation} and that $C_k$ defined in \eqref{Def:Ck}
is a polynomial in $\{G_M^{(j)}(1,t)\}_{j=0}^k$. Therefore, for each $k\geq 1$, if we know the first $k$ moments of $|\mathcal{I}_t|$, then  $C_k$ can be computed without knowing any parameter of the indel process, $t$, nor the initial sequence $\vec{x}$. For example, $C_1= \E_{M}[|\mathcal{I}_t|]$.
\end{remark}

Replacing the expectations $\E_{M}[|\mathcal{I}_t|^k]$ by their corresponding empirical averages from data, we obtain estimators for the unknown initial sequence and the unknown parameters of the indel process from \eqref{ExplictInverses_length}. 
More precisely, the estimators of $\beta$ and $\gamma$, denoted by $\widehat{\beta}$ and $\widehat{\gamma}$, are the analogues of \eqref{betagamma} where $G_M^{(k)}(1,t)$ is replaced by the empirical average of $\mathbb{E}_M\left[\frac{|\mathcal{I}_t|\,!}{(|\mathcal{I}_t|-k)!}1_{|\mathcal{I}_{t}| \geq k}\right]$. Then the estimators of $(M,\,\mu t,\,\lambda t)$ are the empirical analogues of \eqref{ExplictInverses_length} where $(\beta,\gamma)$ is replaced by $(\widehat{\beta},\,\widehat{\gamma})$.

In Section \ref{S:simulation_M}, our simulations demonstrate that these estimators are consistent as the number of independent samples gets larger.

\subsection{Constructive proof for substitution rate and 1-mer}

Let $X_t$ and $Z_t$ be the number of 1's and 0's in $\mathcal{I}_t$ respectively. Then the 1-mer count of $\mathcal{I}_t$ is $(X_t,Z_t)$. 
The \textbf{1-mer count process} $(X_t,Z_t)_{t\geq 0}$ is itself a continuous-time Markov chain with state space $\ZZ_+^2$. More precisely, $\left\{(X_t,Z_t)_{t\geq 0} ,\,\{\P_{(a,b)}\}_{(a,b)\in\ZZ_+^2}\right\}$ is  a Markov family, where  $\mathbb{P}_{(a,b)}$ is the conditional probability measure given that the root state has 1-mer count $(a,b)$.

Theorem \ref{T:Initial_M} implies that we can identify $(|\vec{x}|=M,\,\mu t,\,\lambda t)$ from the first three moments of the length $|\mathcal{I}_t|$ at time $t$. Now if we assume that we know $(|\vec{x}|=M,\,\mu t,\,\lambda t)$, then 
Theorem \ref{T:Initial_1mer} below says that  we can further identify the 1-mer of the ancestral sequence and $\nu t$ from the cross moments (up to second order) of the 1-mer count $(X_t,Z_t)$.

Hence, Theorem \ref{T:Initial_M}  and Theorem \ref{T:Initial_1mer} together imply that we can identify the 1-mer of the ancestral sequence and $(\mu t,\,\lambda t,\,\nu t)$ from the cross moments of $X_t$ and $Z_t$ (up to the second order) and the third moment of $|\mathcal{I}_t|$.

\begin{theorem}[Constructive proof for substitution rate and 1-mer]\label{T:Initial_1mer}
Suppose the length $|\vec{x}|=M\geq 1$ of the initial sequence  and the parameter $\pi_0\in [0,1]$ are known.
For any $t \geq 0$, 
the  mapping 
$\Phi_{t}^{(2)}:\,\{0,1,\cdots,M\}\times (0,\infty)\to \RR_+^5$  defined by 
		\begin{equation}\label{Phi^2}
		\Phi_{t}^{(2)}(a, \nu t)= \Big(\E_{(a, M-a)}[X_t],\,\E_{(a, M-a)}[Z_t],\,\E_{(a, M-a)}[X_t^2],\,\E_{(a, M-a)}[X_tZ_t],\,\E_{(a, M-a)}[Z_t^2] \Big)
		\end{equation}
is one-to-one and its inverse $\left(\Phi_{t}^{(2)}\right)^{-1}$ depends only on 
$(M,\,\mu t,\,\lambda t)$ but not on any other variable including  $\nu$, $t$, and the digits of initial sequence $\vec{x}\in\{0,1\}^M$.
\end{theorem}

The mapping $\Phi^{(2)}$ is determined by the probability generating function (PGF) of the 1-mer count of the sequence. 
The PGF has an explicit expression given in the Appendix. We shall obtain explicit expressions for  $a$ and $\nu t$ in \eqref{eq:aSol} and \eqref{eq:nuSol} below.

\subsection{Constructive proof for ancestral sequence}
Suppose the unknown root sequence is $\vec{x}$ and let $p_{\vec{x}}^{\sigma}(s)$ be the probability that the first digit of the sequence $\mathcal{I}_s$ at time $s$ is $\sigma$. In our setting here the ancestral sequence can be empty, so $p_{\vec{x}}^{0}(s)+p_{\vec{x}}^{1}(s)<1$ for all $s\in(0,\infty)$. 

\begin{theorem}[Constructive proof for ancestral sequence]\label{T:Initial_seq}
Suppose the length $|\vec{x}|=M\geq 1$ of the initial sequence is known. Fix an arbitrary leaf $u\in\partial T$ and let  $t=t_u \in(0,\infty)$ be the (unknown) distance between $u$ and the root of the tree.
Fix an arbitrary increasing sequence of positive times $\{s_j\}_{j=1}^M$ such that $t<s_1<s_2<\cdots <s_M<\infty$. Define the  mapping	$\Phi_{s_1,\cdots,s_{M}}:\,\{0,1\}^M\to [0,1]^{2M}$ by
		\begin{equation}\label{Phi^3}
		\Phi_{s_1,\cdots,s_{M}}(\vec{x})= \left(p^{\sigma}_{\vec{x}}(s_j)\right)_{\sigma \in \{0,1\},\,1\leq j\leq M}.
		\end{equation}
This mapping is one-to-one and has an explicit expression for its inverse that depends only on  $\{(\lambda-\mu) s_j,\,(\mu+\nu) s_j\}_{j=1}^M$ and $\pi_0$, but not on the digits of the initial sequence $\vec{x}\in\{0,1\}^M$.
\end{theorem}

Theorem \ref{T:Initial_seq} is similar in spirit to \cite[Lemma 2]{fan2020statistically}, where the authors seek to identify the initial sequence $\vec{x}$ from the distribution of outcomes for each site in the ending sequence.

By Theorems  \ref{T:Initial_M} and \ref{T:Initial_1mer}, we can identify $\{(\lambda-\mu) s_j,\,(\mu+\nu) s_j\}_{j=1}^M$ provided that  $\pi_0$ is known. Therefore, using Theorem \ref{T:Initial_seq}, we can identify the ancestral sequence $\vec{x}$ from the law $\mathcal{L}_{\vec{x}}(\mathcal{I}_{t_u})$ for any single leaf $u$, if $\pi_0$ is known.

\subsection{Constructive proof for phylogeny}

In this section, we consider the problem of identifying the unknown phylogeny from the distribution of the sequences at the leaves under the indel model. As mentioned below Theorem \ref{T:main}, we do not need to know the value of $\pi_{0}$ to identify the phylogeny.

In ~\cite[equation (27)]{THATTE200658}, the author showed that
under the stationary TKF91 model~\cite{Thorne1991} (i.e.  the root sequence is assumed to follow the stationary distribution), 
for any two leaves $\{u,v\}$, the joint distribution of the sequence lengths $\{L_u, L_v\}$ of the two leaves determines  $\mu\,t_{uv}$, where  $t_{uv}$ is the  distance between the two leaves; see Figure \ref{fig:fork}.  

Here,  we seek to identify $t_{uv}$ using the joint distribution of the sequence lengths \textit{without any assumption} on the root sequence. 
Theorem \ref{thm:phylogeny} below can be viewed as a generalization of the results in \cite[Section 4]{THATTE200658}.
In particular, the full distribution of sequence lengths suffices to recover (identify)  all $\mu t_{uv}$'s for all pairs of leaves. The phylogeny can then be reconstructed from these pairwise distances using standard results; see ~\cite[Chapter 6]{Steel:16} for a general discussion.  We rely specifically on the tree metric results from \cite{zaretskii1965} and \cite{buneman1971}. 

\begin{theorem}[Constructive proof of phylogenetic tree] \label{thm:phylogeny}

Let $T$ be a fixed and unknown phylogeny with at least two   leaves.  The topology and the edge lengths (up to a scaling factor) of $T$ can be identified using only the laws  $\left\{\mathcal{L}_{\vec{x}}\left(|\mathcal{I}_{u}|,\,|\mathcal{I}_{v}| \right)\right\}_{(u,v)\in \partial T\times \partial T}$ of pairwise sequence lengths.
Furthermore, for every pair of leaves $u$ and $v$, their distance $t_{uv}$ satisfy
\begin{align}
    e^{\frac{-(\mu-\lambda)}{2}\,t_{uv}}=&\,
\frac{\mu - \lambda}{M(\lambda +\mu)}\,e^{\frac{+(\mu-\lambda)}{2}\,(t_u+t_v)}\,Cov_M(L^u,L^v)\;+\; e^{\frac{-(\mu-\lambda)}{2}\,(t_u+t_v)}
    \label{Id__dist_1mer_2}\\
    t_w=&\,\frac{t_u+t_v-t_{uv}}{2}, \label{Id__dist_1mer_3}
\end{align}
where $w$ is the most recent common ancestor of $u$ and $v$ in the tree $T$.
\end{theorem}

Let us now explain how the scaled distance $\mu t_{uv}$ can then be identified using Theorem \ref{T:Initial_M}.
Since we know the marginal distribution of $\mathcal{I}_{t_q}$ for all $q\in\partial T$,  Theorem \ref{T:Initial_M} enables us to identify both the length $|\vec{x}|$ of the ancestral sequence and $\{\mu t_q, \lambda t_q, \nu t_q\}$ for all $q\in\partial T$, where $t_q$ is the total distance between the root and leaf $q$. Therefore,
$(\mu-\lambda)t_{uv}$ and $(\mu-\lambda)t_w$ are identified from
\eqref{Id__dist_1mer_2} and \eqref{Id__dist_1mer_3} respectively. From this, $\mu t_{uv}$ and $\mu t_w$ can both be identified using Theorem \ref{T:Initial_M}.

As in the previous Theorem, we take the height $t$ to be the maximum distance to each leaf from the root.  If the tree $T$ is ultrametric, then all the distances are equal.  In the Proofs section, we give a careful proof of the case when the number of leaves is $|\partial T| = 2$.  The $|\partial T|=2$ case then generalizes to an arbitrary number of leaves.  Given sequence lengths at all leaves of a phylogenetic tree, every distance between pairs of leaves is identifiable by this Theorem.  It follows from \cite{zaretskii1965} and \cite{buneman1971} that the tree topology and the branch lengths are identifiable.

\section{Proofs} \label{Proofs}

In the four subsections below, we prove Theorems \ref{T:Initial_M}, \ref{T:Initial_1mer}, \ref{T:Initial_seq} and \ref{thm:phylogeny}. 

\subsection{Identifying $\{|\vec{x}|,\,\lambda t,\,\mu t\}$ jointly using only sequence lengths }\label{S:length}

We first prove Theorem \ref{T:Initial_M} in this subsection.

\begin{prop}\label{L:InvertLength}
Suppose $\lambda \ne \mu$.  Let $(\mathcal{I}_t)_{t\geq 0}$ denote the indel process in Definition \ref{Def_BiINDEL}. Given the distribution of the length $|\mathcal{I}_t|$ at time $t$ starting at an unknown sequence $\vec{x}$, we can identify $|\vec{x}|$ and the parameters $\lambda/\mu$ and $\mu t$ as follows. 
For $k \in \{1, 2, 3\}$, let 
$C_k$ be defined as in
\eqref{Def:Ck}.
Further let $C_2' = C_2/C_1$ and  $C_3' = C_3/ C_1$. Then 
$\gamma = \lambda/\mu$ is given by 
$$\gamma = \frac{\sqrt{-(C_2'+1)^2 \left(3 C_2'^2-2 C_3'\right)}-C_2'^2+C_2'+C_3'}{2 C_2'^2+2 C_2'-C_3'+2},$$
and $\beta_t = e^{-(\mu - \lambda)t}$ is given by 
$$\beta_t = \frac{\gamma(2 + C_2') - C_2'}{1 + \gamma}.$$
\end{prop}

\begin{proof}[Proof of Proposition \ref{L:InvertLength}]

Recall from  \cite[Appendix]{fan2020statistically} that  the probability generating functions \[G_M(z,t):=\sum_{j=0}^{\infty}p_{Mj}(t)z^j=\mathbb{E}_M[z^{|\mathcal{I}_t|}]\] 
of the length process are given by
\begin{equation}\label{GenFcn_length}
G_M(z,t)=\left[\frac{1-\beta_t -z(\gamma-\beta_t)}{1-\beta_t\gamma -\gamma z(1-\beta_t)}\right]^M
\end{equation}
for  $M\in \ZZ_+$, $t>0$ and $z\in (0,\,\mu/\lambda)$,
where 
\begin{equation*}
\beta_t:= e^{-(\mu-\lambda)t},\quad \gamma:=\frac{\lambda}{\mu}.
\end{equation*}

Write $\beta = \beta_t$ for simplicity.
Taking the logarithm of \eqref{GenFcn_length} gives a linear function in $M$, and 
we then differentiate both sides with respect to $z$ to get rational functions in $z$. These give
 the following system of three equations: 
\begin{eqnarray*}
C_1=& \beta M\\
C_2=&  \frac{\beta(2 \gamma - \beta \gamma - \beta)}{1 - \gamma} M \\
C_3=& \frac{2 \beta(\beta^2 \gamma^2 + \beta^2 \gamma + \beta^2 - 3 \beta \gamma^2 - 3 \beta \gamma + 3 \gamma^2) } {(1 - \gamma)^2} M.
\end{eqnarray*}
After substituting in $C_1 = \beta M$, the second equation can be used to express $\beta$ in terms of $\gamma$ and $C_2'$:
$$\beta = \frac{\gamma(2 + C_2') - C_2'}{1 + \gamma}.$$
Substituting $\beta$ into the third equation yields a quadratic in $\gamma$:
$$\gamma^2(2C_2'^2 + 2C_2' -C_3'+ 2) + 2\gamma(C_2'^2 - C_2' - C_3') + 2C_2'^2 - C_3' = 0.$$
This quadratic must have at least one solution $\gamma > 0$ whose corresponding $\beta$ is greater than $0$ since those are the parameters for the indel process. If the quadratic has only one unique root, then there is no issue with identifying $\gamma$ and $\beta.$ Hence, we consider the case that the quadratic has two distinct roots.
Let $\gamma'$ be the root of the quadratic that is not equal to $\gamma$. Let $\beta'$ correspond to $\gamma'$, i.e.
$$\beta' = \frac{\gamma'(2 + C_2') - C_2'}{1 + \gamma'}.$$
We claim that either $\gamma' \le 0$ or $\beta' \le 0$, thereby showing that the parameters $\{\gamma, \beta\}$ can be identified. Assume that $\gamma' > 0$ and $\beta' > 0$ for the sake of contradiction.
Using the sum of the roots of the quadratic, we have
$$\gamma' = \frac{-2(C_2'^2 - C_2' - C_3')}{2 C_2'^2 + 2C_2' - C_3' + 2} - \gamma = \frac{\beta - \gamma}{\beta\gamma - 1}.$$
Since $\beta' > 0$, we must have $\gamma'(2 + C_2') - C_2' > 0$, which is
$-\frac{(\beta-1)\beta(1 + \gamma)}{\beta \gamma-1} > 0.$ This occurs if and only if 
$\frac{1 - \beta }{\beta \gamma-1} > 0$.

Recalling that $\beta = e^{-(\mu - \lambda)t}$ and $\gamma = \lambda / \mu$, we see that $\beta - 1$ and $\gamma - 1$ have the same sign. Therefore, $(1 - \beta) / (\beta \gamma -1)$ must be less than or equal to 0.

Next, we show that $\gamma$ is equal to the largest root $\gamma''$ of the quadratic. Indeed, since $\beta > 0$, we deduce $\gamma(2+ C_2') - C_2' > 0$, and hence $\gamma''(2 + C_2') - C_2' > 0$. But then $\beta'' > 0$. From the above, we conclude $\gamma = \gamma''.$ Thus,
\begin{align*}
    \gamma = \gamma'' &= 
\frac{ -2(C_2'^2 - C_2' - C_3') + \sqrt{4(C_2'^2-C_2'-C_3')^2 - 4(2C_2'^2 + 2C_2' - C_3'+2)(2C_2'^2-C_3')}}{2(2C_2'^2+2C_2'-C_3'+2)} \\&= 
\frac{\sqrt{-(C_2'+1)^2 \left(3 C_2'^2-2 C_3'\right)}-C_2'^2+C_2'+C_3'}{2 C_2'^2+2 C_2'-C_3'+2},
\end{align*}
where we used the fact that 
$$ 2C_2'^2 + 2 C_2' - C_3' + 2 = \frac { 2(\beta - 1)(\beta \gamma - 1)}{(\gamma - 1)^2}$$
is always positive.
\end{proof}

We can now prove Theorem \ref{T:Initial_M}.

\begin{proof}[Proof of Theorem \ref{T:Initial_M}]
Proposition \ref{L:InvertLength} 
identifies $\{e^{-(\mu-\lambda)t},\, \gamma:=\lambda/\mu\}$, and therefore $\{\mu t, \lambda t\}$, through explicit expressions. 
\end{proof}

\subsection{Identifying $\nu t$ and the 1-mer of $\vec{x}$, given $\{|\vec{x}|,\lambda t,\mu t,\,\pi_0\}$}

The $1$-mer of a sequence is a vector in $\mathbb{Z}_+^2$ whose first and second entries are the numbers of 1's and 0's in the sequence. In \cite{allman2017statistically}, $k$-mers have been used in phylogeny inference.

Suppose $(a,b)\in \mathbb{Z}_+^2$ is the 1-mer of the initial sequence $\vec{x}$, both of which we do not know. However we know the length of $\vec{x}$ is $M$ because $M$ is identified (by our results in Section \ref{S:length}), i.e. we know that $|\vec{x}|=a+b=M$. Therefore, the question here is about identifying $\{\nu t,\,a\}$. 

A key advantage of using the vector of $1$-mer counts under TKF91 is that it still follows a Markov process.  The probability  generating function of $(X_t,Z_t)$ under $\P_{(a,b)}$ is defined by
\[G(z_1,z_2,\,t):=G_{(a,b)}(z_1,z_2,\,t):=\mathbb{E}_{(a,b)}[z_1^{X_t}\,z_2^{Z_t}]\]
for all $z_1,\,z_2\in \RR$ such that the right hand side exists in $\RR$. 
The transition rate $Q_{(i,j),(k,\ell)}$ from $(i,j)$ to $(k,\ell)$ is given by, for all $(i,j)\in \ZZ_+^2$,
\begin{equation}\label{imer_rates}
    Q_{(i,j),(k,\ell)}=
    \begin{cases}
    \lambda(i+j) \pi_1\quad &\text{if }(k,\ell)=(i+1,j) \qquad \qquad \text{a 1 is inserted}\\
    \mu i \quad & \text{if }(k,\ell)=(i-1,j)  \qquad \qquad \text{a 1 is deleted}\\
    \lambda(i+j) \pi_0 \quad & \text{if }(k,\ell)=(i,j+1) \qquad \qquad\text{a 0 is inserted}\\
    \mu j \quad & \text{if }(k,\ell)=(i,j-1) \qquad \qquad \text{a 0 is deleted}\\
    \nu\pi_1\,j \quad & \text{if }(k,\ell)=(i+1,j-1) \qquad \text{ a 0 is switched to 1}\\
    \nu\pi_0\,i \quad & \text{if }(k,\ell)=(i-1,j+1)\qquad \text{ a 1 is switched to 0}\\
    -(\lambda+\mu+\nu)(i+j) \quad & \text{if }(k,\ell)=(i,j)\\  
    0 \quad & \text{otherwise } 
    \end{cases}
\end{equation}

\medskip

Let $G_{(a,b)}(z_1,z_2,t) := \mathbb{E}_{(a,b)}[z_1^{X_t}z_2^{Z_t}]$ be the probability generating function. As in Remark \ref{Rk:Ck}, the partial derivatives of $G_{(a,b)}$ 
and the moments of $(X_t,Z_t)$ are related as follows:
\begin{equation} \label{eqn:Gab-11t}
     E_{i,j} := \frac{\partial^{(i+j)} G_{(a,b)}(z_1,z_2,t)}{\partial^{(i)} z_1\,\partial^{(j)} z_2} \bigg|_{z_1 = z_2 = 1} = \E_{(a,b)} \left[ \frac{|X_t|!}{(|X_t| - i)!}\, \frac{|Z_t|!}{(|Z_t| - j)!}\,1_{\{|X_t| \geq i,\,|Z_t|\geq j\}} \right].
 \end{equation} 

The statement of Proposition \ref{prop:1merpi0} says that $\nu t$ and $a$ are identifiable, assuming other parameters are estimable or known in advance.  

\medskip

\begin{prop}[PGF of 1-mer] \label{prop:1mer}
The probability generating function 
$$G_{(a,b)}(z_1, z_2; t):= \mathbb{E}_{(a,b)}[z_1^{X_t}z_2^{Z_t}]$$ 
for the 1-mer at time $t$ is an explicit function in $(a,b,\,z_1,z_2,\,\mu t, \lambda t, \nu t, \pi)$,
given by
$$G_{(a,b)}(z_1, z_2; t) = r_1^a r_2^b,$$
where, under the  substitutions
$y_1 = z_1 - z_2$ and $y_2 = \pi_1 z_1 + \pi_0 z_2$, 
\begin{align}
    r_1 &= \frac{-\mu  e^{\lambda t}+\mu  e^{\mu t}+y_2 \left(\mu 
   e^{\lambda t}-\lambda  e^{\mu t}\right)+\pi _0 y_1 (\mu -\lambda )e^{-\nu t}}
  {\left(\mu -\lambda  y_2\right)e^{\mu  t}+\lambda  \left(y_2-1\right) e^{\lambda 
   t}}
   \\
   r_2 &= \frac{-\mu e^{\lambda  t}+\mu  e^{\mu  t}+ y_2(\mu e^{\lambda  t}-\lambda   e^{\mu  t}) +\pi_1 y_1 (\lambda -\mu ) e^{-\nu t}}
   { \left(\mu -\lambda  y_2\right)e^{\mu  t}+\lambda  \left(y_2-1\right) e^{\lambda 
   t}}.
\end{align}
\end{prop}

Proposition \ref{prop:1mer} is proved in the Appendix, in Section~\ref{s:1-merAppendix}. In view of this proposition, 
there is a function $H_{(a,b)}(y_1,y_2;t)$ determined by considering $r_1^a r_2^b$ as a function of the independent variables $y_1$ and $y_2$ without consideration to their dependencies on $z_1$ and $z_2$.  As the mapping $(y_1,y_2) = (z_1-z_2,\pi_1 z_{1}+\pi_0z_2)$ is an invertible linear mapping, the equality of two hypothetical $G^{(1)}$ and $G^{(2)}$ induced by two parameter settings implies the equality of their respective $H^{(1)}$ and $H^{(2)}$.  It would remain to show that the parameters are uniquely determined by $H$.  Observe that the right-hand side of \eqref{eqn:Gab-11t} is equal to $H_{(a,b)}(0,1;t)$.

For completeness, we give an identifiability result with an explicit estimation method when $\pi_0$ is treated as known. This assumption is necessary for analytical reasons.  While it may be conjectured  that the parameters $(\pi_0, a, \nu t)$ are determined uniquely from $1$-mer count vectors at time $t$ and the known parameters, the solution is determined by a system of non-linear polynomial equations that we found analytically intractable; these equations are not needed in this paper, but we recorded them in the Appendix of \cite[Section \ref{sss:1-mer System}]{xue2024joint} for future research.

On the other hand, it is not unusual for the relative proportion of nucleotide labels to be known in substitution-only models including the F81, HKY, and GTR models; see \cite{felsenstein1981,hasegawa1985,tavare1986}.  The relative proportions might be explicitly assumed to satisfy equality constraints such as the F81, or the relative proportions are viewed as ``free parameters'' which can be determined experimentally.  In view of this information, it is reasonable to try to infer the rate of substitution $\nu t$.  The problem suggested by Proposition \ref{prop:1merpi0} is a generalization of this idea to multiple mutation-related parameters.  A notable weakness to this approach is that the experimental inference of substitution probabilities from relative proportions is assumed to be at or near stationarity.

\begin{prop} \label{prop:1merpi0}
    Suppose we know the values of $\lambda t \in (0,\infty)$, $\mu t \in (0,\infty)$, the ancestral sequence length $|\vec{x}| = M \in \mathbb{N}$, and the parameter $\pi_0 \in [0,1]$.  Then $\nu t$ and $a$ are identifiable, with values determined explicitly from $\{E_{1,0},E_{0,1},E_{2,0}, E_{1,1}, E_{0,2}\}$ and the known parameters $\{M=|\vec{x}|,\,pi_0,\,\mu t,\,\lambda t\}$.
\end{prop}

\begin{proof}
For $i.j \in \ZZ_+$, we let $H_{(a,b)}(y_1,y_2;t)=G_{(a,b)}(z_1,z_2,t)$ where $(y_1,y_2) = (z_1-z_2,\pi_1 z_{1}+\pi_0z_2)$.
By chain rule,  
\begin{align*}
        \frac{\partial H_{(a,b)}(0, 1;t)}{\partial y_1} &= \pi_0 \frac{\partial G_{(a,b)}(1,1;t)}{\partial z_1} - \pi_1 \frac{\partial G_{(a,b)}(1,1;t)}{\partial z_2} \\
        \frac{\partial^2 H_{(a,b)}(0, 1;t)}{\partial y_1^2} &= \pi_0^2 \frac{\partial^2 G_{(a,b)}(1, 1;t)}{\partial z_1^2} - 2\pi_0 \pi_1 \frac{\partial^2 G_{(a,b)}(1, 1;t)}{\partial z_1 \partial z_2} + \pi_1^2 \frac{\partial^2 G_{(a,b)}(1, 1;t)}{\partial z_2^2} \\
        \frac{\partial^3 H_{(a,b)}(0, 1;t)}{\partial y_1^3} &= \pi_0^3 \frac{\partial^3 G_{(a,b)}(1, 1;t)}{\partial z_1^3} - 3\pi_0^2 \pi_1 \frac{\partial^3 G_{(a,b)}(1, 1;t)}{\partial z_1^2 \partial z_2} \\
        &+ 3\pi_0 \pi_1^2 \frac{\partial^3 G_{(a,b)}(1, 1;t)}{\partial z_1 \partial z_2^2} - \pi_1^3 \frac{\partial^3 G_{(a,b)}(1, 1;t)}{\partial z_1^3}.
    \end{align*}
The left-hand side simplifies as 
\begin{align}
\frac{\partial H_{(a,b)}(0, 1;t)}{\partial y_1} =\frac{\partial \ln H_{(a,b)}(0, 1;t)}{\partial y_1} =& a\pi_0 e^{-\nu t} e^{-\mu t} - (M - a)\pi_1 e^{-\nu t}e^{-\mu t} \notag\\
=& ae^{-\nu t} e^{-\mu t} - M \pi_1 e^{-\nu t} e^{-\mu t}   
\end{align}
and \begin{align*}\frac{\partial^2 H_{(a,b)}(0, 1;t)}{\partial y_1^2} &= \frac{\partial^2 \ln H_{(a,b)}(0, 1;t)}{\partial^2 y_1} + \left(\frac{\partial H_{(a,b)}(0, 1;t)}{\partial y_1}\right)^2 \\
&= -a \pi_0^2 e^{-2\nu t}e^{-2\mu t} - (M - a)\pi_1^2 e^{-2\nu t}e^{-2\mu t} + \left(\pi_0 E_{1,0} - \pi_1 E_{0,1}\right)^2
\end{align*}

The first equation asserts that 
\begin{equation}\label{eq:solvea}
ae^{-\nu t} - M \pi_1 e^{-\nu t} = e^{\mu t}(\pi_0 E_{1,0} - \pi_1 E_{0,1}),    
\end{equation}
where we recalled \eqref{eqn:Gab-11t}.

The second equation asserts that 
\begin{align*} 
    ae^{-2\nu t}(\pi_1^2 - \pi_0^2) - M \pi_1^2 e^{-2\nu t} = e^{2\mu t}\left[\pi_0^2 E_{2,0} - 2\pi_0\pi_1 E_{1,1} + \pi_1^2 E_{0,2} - \left(\pi_0 E_{1,0} - \pi_1 E_{0,1}\right)^2\right]
\end{align*}
which gives
\begin{align}\label{eq:secondSub}
    &e^{-2\nu t}(-M\pi_1^2 + M\pi_1 (\pi_1^2 - \pi_0^2)) + e^{-\nu t} e^{\mu t}(\pi_1^2 - \pi_0^2)(\pi_0 E_{1,0} - \pi_1 E_{0,1}) \notag\\
    &- e^{2\mu t}\left[\pi_0^2 E_{2,0} - 2\pi_0\pi_1 E_{1,1} + \pi_1^2 E_{0,2} - \left(\pi_0 E_{1,0} - \pi_1 E_{0,1}\right)^2\right] = 0.
\end{align} 
Equation \eqref{eq:secondSub} is a quadratic equation for $e^{-\nu t}$, of the form $Ae^{-2\nu t}+Be^{-\nu t}+C=0$ where  the leading term is
$$A=-M\pi_1^2 + M\pi_1 (\pi_1^2 - \pi_0^2) <0$$  
and the last term  is
$$C=- e^{2\mu t}\left[\pi_0^2 E_{2,0} - 2\pi_0\pi_1 E_{1,1} + \pi_1^2 E_{0,2} - \left(\pi_0 E_{1,0} - \pi_1 E_{0,1}\right)^2\right]>0.$$
Regardless of the sign of the middle term, there is exactly one sign change. Therefore, by Descartes' rule of signs, this quadratic  equation has exactly one positive root which is $\frac{-B-\sqrt{B^2-4AC}}{2A}$. That is,
\begin{align}
    e^{-\nu t} &= \frac{1}{2(-M\pi_1^2 +M\pi_1(\pi_1^2 - \pi_0^2)} \bigg[ -e^{\mu t} (\pi_1^2-\pi_0^2)(\pi_0 E_{1,0} - \pi_1 E_{0,1}) \notag \\
    &- \bigg(e^{2\mu t} (\pi_1^2-\pi_0^2)^2(\pi_0 E_{1,0} - \pi_1 E_{0,1})^2 \notag \\
    \label{nu_sol}
    &+ 4(-M\pi_1^2 + M\pi_1(\pi_1^2 - \pi_0^2)e^{2\mu t} \left[\pi_0^2 E_{2,0} - 2\pi_0 \pi_1 E_{1,1} + \pi_1^2 E_{0,2} - (\pi_0 E_{1,0} - \pi_1 E_{0,1})^2 \right]\bigg)^{1/2} \bigg]
\end{align}

We then plug the solution \eqref{nu_sol} into \eqref{eq:solvea} to solve for $a$, obtaining that
\begin{equation} \label{eq:aSol}
    a= M\pi_1 + e^{(\nu + \mu)t}(\pi_0 E_{1,0} - \pi_1 E_{0,1}).
\end{equation} Further, equation \eqref{eq:secondSub} has a simple solution when  $\pi_0 = \pi_1 = 1/2$.  It has $A = -M\pi_1^2$, $B = 0$, and $C$ as before.  We then have \begin{align}
    e^{-\nu t} &= -\sqrt{-\frac{C}{A}} \nonumber \\
    &= -\sqrt{\frac{- e^{2\mu t}\left[\pi_0^2 E_{2,0} - 2\pi_0\pi_1 E_{1,1} + \pi_1^2 E_{0,2} - \left(\pi_0 E_{1,0} - \pi_1 E_{0,1}\right)^2\right]}{-M\pi_1^2 + M\pi_1(\pi_1^2 - \pi_0^2)}}. \label{eq:nuSol}
\end{align}
Therefore, $\nu t$ and $a$ are identified in \eqref{eq:nuSol} and \eqref{eq:aSol}.
\end{proof} 

We are now ready to prove Theorem \ref{T:Initial_1mer}.

\begin{proof}[Proof of Theorem \ref{T:Initial_1mer}]
    In the statement of the Theorem, let the first- and second-order raw moments of $X_t$ and $Z_t$ be given.  The numbers $E_{i,j}$ are determined by these moments by Proposition \ref{prop:1merpi0}.  From these and the known values, the expressions in \eqref{eq:aSol} and \eqref{eq:nuSol} are then the unique expressions for $a$ and $e^{-\nu t}$.
\end{proof}

We end this subsection with a remark that helps explain why we assumed that $\pi_0$ is known rather than jointly identifying the triple $\{\pi_0,\,\nu t,\,a\}$. 

\begin{remark}\rm[Impossibility of joint identification for $(\pi_0, \nu t,  \vec{x})$]\label{Rk:Impossibility_pi0}
It is impossible to identify $(\pi_0, \nu t, \vec{x})$ on $(0,1)\times (0,\infty) \times \{0,1\}$ using only the distribution of the indel process at a single time $t\geq 0$.  This can readily be seen from the $2$-state substitution-only process (i.e. $\mu=\lambda=0$ and $|\vec{x}|=1$) as follows:
 the map $\Phi:\,(0,1)\times (0,\infty) \times\{0,1\} \to [0,1]^2$ defined by
\[
(\pi_0,\,\nu t,\,0) \mapsto \left(\mathbb{P}_0(\mathcal{I}_t=0) ,\;\mathbb{P}_0(\mathcal{I}_t=1) \right)=
\left(\pi_0+ \pi_1e^{-\nu t} ,\;\pi_1(1 - e^{-\nu t})\right)
\]
and
\[
(\pi_0,\,\nu t,\,1) \mapsto \left(\mathbb{P}_1(\mathcal{I}_t=0) ,\;\mathbb{P}_1(\mathcal{I}_t=1) \right)=
\left(\pi_0(1-e^{-\nu t}) ,\;\pi_1 + \pi_0 e^{-\nu t} \right)
\]
is not injective.
\end{remark}

\subsection{Identifying $\vec{x}$ given  $\{|\vec{x}|,\lambda t,\mu t,\nu t,\,\pi_0\}$}

Theorem \ref{T:Initial_seq} gives a constructive proof of initial-state identifiability, as long as we know $\pi_0$, from the law of the indel process at time $t$ and from previously identified parameters $\{|\vec x|, \lambda t, \mu t, \nu t\}$.  The proof of Theorem \ref{T:Initial_seq} is simialr to that of \cite[Lemma 2]{fan2020statistically}.

\begin{proof}[Proof of Theorem \ref{T:Initial_seq}]
For any time $t$, let $\beta_{t} = e^{(\lambda - \mu)t}$, and let
\begin{gather*} 
    \eta(t) =
\frac{1- \beta_{t}}{1-\gamma \beta_{t}}= \P(|\mathcal{I}_{t}|=0\;\big|\;|\mathcal{I}_0|=1), \\ 
\psi(t) = e^{-(\mu+\nu)t}, \quad \phi_j=\frac{
-\psi(t)+ 1 - \eta(t)}{1-\eta(t)}.
\end{gather*}
Let 
$$ \eta_j = \eta(s_j), \quad \psi_j = \psi(s_j), \quad \phi_j = \phi(s_j),$$
where $t$ here refers to the (unknown) height of the tree. By Lemma~\ref{L:theta_v} in the Appendix,
\begin{equation}\label{E:Solve for x}
p^{\sigma}_{\vec{x}}(s_j)= \pi_{\sigma}\,\phi_j\,\big[1-\eta_j^{M}\big]\,+\,\psi_j\sum_{i=1}^{M}1_{\{x_i=\sigma\}}\,\eta_j^{i-1},
\end{equation}
for all $\sigma\in \{0,1\}$  and $1\leq j\leq M$,
where $M = |\vec{x}|$. 

Our goal is to solve~\eqref{E:Solve for x} for  the unknown vector $\vec{x}=(x_i)\in \{0,1\}^M$.
System~\eqref{E:Solve for x} is equivalent to the matrix equation 
	\begin{equation}\label{E:Solve for x2}
	U = \Psi\,V\, Y^{\vec{x}},
	\end{equation}
	where 
	\begin{enumerate}
		\item 	$Y^{\vec{x}}$ is the unknown $M\times 2$ matrix whose entries are  $1_{\{x_j=\sigma\}}$, and
		\item $\Psi$ is the $M\times M$ diagonal matrix whose diagonal entries are $\{\psi_j\}_{j=1}^{M}$,
		\item $U$ is the $M\times 2$ matrix with entries 
		\begin{equation*}
		U_{j, \sigma}=p^{\sigma}_{\vec{x}}(s_j)- \pi_{\sigma}\,\phi_j\,\big[1-\eta_j^{M}\big].
		\end{equation*}
		\item $V$ is the $M\times M$ Vandermonde matrix 
		\begin{equation*}
		V=
		\begin{pmatrix}
		1&\eta_1&\eta_1^2&\dots&\eta_1^{M-1}\\
		1&\eta_2&\eta_2^2&\dots&\eta_2^{M-1}\\
		&&\vdots\\
		1&\eta_{M}&\eta_{M}^2&\dots &\eta_{M}^{M-1}\\
		\end{pmatrix}.
		\end{equation*}
	\end{enumerate}
	It is well-known that the Vandermonde matrix $V$ is invertible (see, e.g.,~\cite[Theorem 1]{gautschi1962inverses}), so we can solve the system~\eqref{E:Solve for x2} to obtain
		\begin{equation}\label{E:Solve for x3}
		 Y^{\vec{x}} = V^{-1}\,\Psi^{-1}\,U.
		\end{equation}
	Sequence $\vec{x}\in \{0,1\}^{M}$ is uniquely determined by $ Y^{\vec{x}} $. Hence, from~\eqref{E:Solve for x3}, we get an explicit inverse for the mappings $\Phi_{s_1,\cdots,s_M}$ defined in~\eqref{Phi^3}. 
\end{proof}

\subsection{Identifiability of the pairwise distance $t_{uv}$}
\label{sec:pairwise}

Fix an arbitrary pair of leaves $\{u,v\}$ and let $w$ be the most recent common ancestor of $u$ and $v$. Then $t_{uv}=t_{wu}+t_{wv}$, where $t_{uv}$ is the distance between $u$ and $v$. See Figure \ref{fig:fork}.
\FloatBarrier
\begin{figure}[h!]
    \centering
    \includegraphics[scale=0.5]{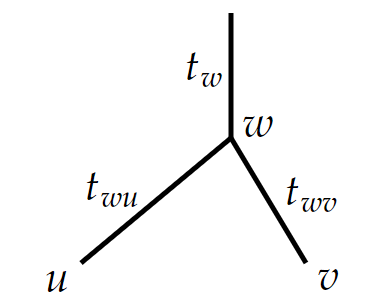}
    \caption{The subtree of the phylogeny $T$ with two leaves $\{u,v\}$ and the root, where $w$ is the most recent common ancestor of $u$ and $v$.}
    \label{fig:fork}
\end{figure}
\FloatBarrier

Let $t_a$ be the distance of an arbitrary node $a\in T$ from the root, $L^a$ be the sequence length at $a$, and $(L^a-Z^a, Z^a)$ be the 1-mer (number of 1's and 0's) at $a$. 
Our results in Section \ref{S:length} imply that  $\{\vec{x},\,\mu t_u,\lambda t_u, \nu t_u, \pi_0 \}$ can be identified for each leaf $u$. However, this does not yet tell us what $\mu t_{uv}$ is. For this, we need to use the \textit{correlation} between the two leaf sequences. 

Let $M=|\vec{x}|$ be the length of the ancestral sequence and  $(L^u,\,L^v)$ be the sequence lengths at $u$ and $v$.  The joint distribution of $(L^u,\,L^v)$ under $\mathbb{P}_M$ can be computed as follows. Applying the Markov property at time $t_{w}$ followed by conditional independence, we obtain that for all $(i,j)\in \mathbb{Z}_+^2$ that
\begin{equation}\label{Length_jointpmf}
\mathbb{P}_M\left(\,(L^u,\,L^v)=(i,j)\,\right)=\sum_{k\in \ZZ_+}p^{t_w}_{M}(k)\,p^{t_u-t_w}_{k}(i)\,p^{t_v-t_w}_{k}(j),
\end{equation}
where  $p^s_{M}(k):=\mathbb{P}_M\left(L_s=k\right)$ is the probability mass function of the length $L_s$ of a single sequence under the indel process at time $s$.

We now identify $\mu t_{uv}$ by computing the covariance of $(L^u,\,L^v)$ under $\P_{M}$, where $M$ is an arbitrary positive integer. This covariance is
\begin{equation}\label{Def:Cov_length}
Cov_M(L^u,L^v):= \mathbb{E}_{M}[ (L^u-m^u)(L^v-m^v)],
\end{equation}
where $m^u=\mathbb{E}_{M}[L^u]=M e^{-(\mu-\lambda)t_u}$ and $m^v=\mathbb{E}_{M}[L^v]=M e^{-(\mu-\lambda)t_v}$ are the means.  Now, we prove Theorem \ref{thm:phylogeny}.

\begin{proof}[Proof of Theorem \ref{thm:phylogeny}]
By our results in Section \ref{S:length},
both $t_u=t_w+t_{wu}$ and $t_v=t_w+t_{wv}$ are identified. Adding these up gives  \eqref{Id__dist_1mer_3}.

Applying the Markov property at time $t_{w}$ and then conditional independence, we obtain
\begin{align*}
&Cov_M(L^u,L^v):=\, \mathbb{E}_{M}[ (L^u-m^u)(L^v-m^v)]\\
=&\,\sum_{k\in \ZZ_+}p^{t_w}_{M}(k)\,\mathbb{E}_{k}\left[ L_{t_u-t_w}-Me^{-(\mu-\lambda)t_u}\right]\,\mathbb{E}_{k}\left[L_{t_v-t_w}-Me^{-(\mu-\lambda)t_v}\right]\\
=&\,\sum_{k\in \ZZ_+}p^{t_w}_{M}(k)\,\left( ke^{-(\mu-\lambda)(t_u-t_w)} -M e^{-(\mu-\lambda)t_u}\right)\,\left(ke^{-(\mu-\lambda)(t_v-t_w)}-M e^{-(\mu-\lambda)t_v}\right)\\
=&\,\sum_{k\in \ZZ_+}p^{t_w}_{M}(k)\,\left(
k^2\,e^{-(\mu-\lambda)t_{uv}} - 2k M\,e^{-(\mu-\lambda)(t_{uv}+t_w)}+M^2
e^{-(\mu-\lambda)(t_{u}+t_{v})}
\right)  \label{cov_length_1}\\
=&\,e^{-(\mu-\lambda)t_{uv}}\,\E_M[L^2_{t_w}]\,-\,2M\,e^{-(\mu-\lambda)(t_{uv}+t_w)}\,\E_M[L_{t_w}]\,+\,M^2
e^{-(\mu-\lambda)(t_{u}+t_{v})},
\end{align*}
where we also used the fact $t_{uv}=t_u+t_v-2t_{w}$ (or, equivalently, \eqref{Id__dist_1mer_3}).

The second moment $m_2(t):=\mathbb {E}_{M}[(L_t)^2]$ satisfies
$$\frac{d m_2}{dt} = 
2(\lambda - \mu)m_2 + (\lambda + \mu) \mathbb {E}_{M}[L_t],$$ with initial condition
$m_2(0) = M^2$.
Hence, $$m_2 = \left(M^2 - \frac{M(\lambda + \mu)}{\mu - \lambda} \right) e^{-2(\mu-\lambda)t} + \frac{M(\lambda+\mu)}{\mu-\lambda}e^{-(\mu-\lambda)t}.$$

Hence \eqref{Id__dist_1mer_2} can be obtained:
\begin{align*}
&Cov_M(L^u,L^v) \\
=&\,M^2\left[
e^{-(\mu-\lambda)t_{uv}}\,e^{-2(\mu-\lambda)t_w}-2e^{-(\mu-\lambda)(t_{uv}+t_w)}\,e^{-(\mu-\lambda)t_w}+e^{-(\mu-\lambda)(t_{u}+t_{v})}
\right] \\ &\quad + e^{-(\mu-\lambda) t_{uv}}\frac{M(\lambda+\mu)}{\mu-\lambda}\left(-e^{-2(\mu-\lambda) t_w} + e^{-(\mu-\lambda)t_w}\right)\\
=&\,M^2\left[
-e^{-(\mu-\lambda)(t_{uv}+2t_w)}+e^{-(\mu-\lambda)(t_{u}+t_{v})}
\right] \\&\quad + e^{-(\mu-\lambda) t_{uv}}\frac{M(\lambda+\mu)}{\mu-\lambda}\left(-e^{-2(\mu-\lambda) t_w} + e^{-(\mu-\lambda)t_w}\right)\\
=&\, e^{-(\mu-\lambda) t_{uv}}\frac{M(\lambda+\mu)}{\mu-\lambda}\left(-e^{-2(\mu-\lambda) t_w} + e^{-(\mu-\lambda)t_w}\right)\\
=&\, \frac{M(\lambda+\mu)}{\mu-\lambda}\left(-e^{-(\mu-\lambda) (t_u+t_v)} + e^{-(\mu-\lambda)(t_{uv}+t_w)}\right).
\end{align*}

Solving for $t_{uv}$ in the above expression and using \eqref{Id__dist_1mer_3}, we obtain
\[
   e^{\frac{-(\mu-\lambda)}{2}\,t_{uv}}=\,
\frac{Cov_M(L^u,L^v)+ \frac{M(\lambda+\mu)}{\mu-\lambda}\,e^{-(\mu-\lambda) (t_u+t_v)}}{\frac{M(\lambda+\mu)}{\mu-\lambda}\,e^{\frac{-(\mu-\lambda)}{2}\,(t_u+t_v)}}
\]
and hence formula \eqref{Id__dist_1mer_2}.
\end{proof}

\section{Algorithms and simulation studies}\label{S:Simulation}
In this section, we illustrate applications of our joint identifications by introducing and implementing reconstruction algorithms.

The set-up for Sections~\ref{S:simulation_M} to \ref{S:simulation_x} is as follows. Fix an arbitrary leaf $u \in\partial T$ of the phylogeny and let $t=t_u$ be its distance from the root. We consider $N$ i.i.d samples of $\mathcal{I}_{t}$ and use them to reconstruct the indel parameters and the ancestral sequence $\vec{x}$ (assuming $\pi_0$ is known in the latter).


In Section \ref{S:simulation_M}, we reconstruct the length $M=|\vec{x}|$ of the ancestral sequence and the parameters $\gamma=\lambda/\mu$ and $\beta= e^{-(\mu - \lambda) t}$ using only the lengths of the $N$ sampled sequences. In Section \ref{S:simulation_1mer}, we use $M, \lambda t, \pi_0$, and the 1-mers of the sequences to reconstruct the ancestral 1-mer and $\nu t.$  In Section \ref{S:simulation_x}, we use $M, \lambda t, \mu t, \nu t, \pi_0$ and the sampled sequences to reconstruct the ancestral sequence.
Lastly, in Section \ref{S:simulation_T}, we use a different set-up. We instead consider $N$ i.i.d samples of a fork tree with two leaves $u$ and $v$ (see Figure \ref{fig:fork}). We use the joint distribution of the lengths of the sequences at $u$ and $v$ to reconstruct the distance between them.

With the exception of the algorithm in Section \ref{S:simulation_x}, the other three algorithms are efficient and very straightforward to implement, as they essentially use only basic for-loops.  
The code used to perform the simulations and to generate the figures is available at \url{https://github.com/alexxue99/Invertibility}.

\subsection{Reconstruction of ancestral sequence length and indel rates}\label{S:simulation_M}

In this section, we give an algorithm to reconstruct three quantities, namely the ancestral sequence length $M$, $\gamma = \lambda/ \mu$ and $\beta = e^{-(\mu - \lambda) t}$,  using only the lengths $l_1, l_2, \ldots, l_N$ of the $N$ sampled sequences. The algorithm uses our formulas in Proposition \ref{L:InvertLength}.

\FloatBarrier
\begin{algorithm}
\caption{Ancestral sequence length and indel rates reconstruction} \label{length:alg}
\begin{algorithmic} 
\State \textbf{Input}: $N$, $l_1, l_2, \dots, l_N$
\\
\For {$k = 1$ to $3$}\State
 $ G_M^{(k)}(1, t) \gets \frac{1}{N}\sum_{i = 1}^N \frac{l_i\,!}{(l_i-k)!}1_{l_i \ge k} $
 \EndFor 
 \\ 
 \State $ C_1 \gets  G_M^{(1)}(1, t)  $
\State $ C_2 \gets  G_M^{(2)}(1, t) - ( G_M^{(1)}(1, t))^2   $
\State $ C_3 \gets  G_M^{(3)}(1, t) + 2( G_M^{(1)}(1, t))^3 - 3 G_M^{(1)}(1, t)  G_M^{(2)}(1, t)$ \\
\State$ C_2' \gets  C_2 /  C_1 $ 
\State$ C_3' \gets  C_3 /  C_1 $\\
\State$ \gamma \gets \frac{\sqrt{-( C_2'+1)^2 \left(3  C_2'^2-2  C_3'\right)}- C_2'^2+ C_2'+ C_3'}{2  C_2'^2+2  C_2'- C_3'+2}  $
\State$ \beta \gets \frac{ \gamma(2 +  C_2') -  C_2'}{1 +  \gamma} $
\State$ M \gets  C_1 /  \beta$
\\
\State \textbf{return} $\gamma, \beta, M.$
\end{algorithmic}
\end{algorithm}
\FloatBarrier

To experimentally verify that Algorithm \ref{length:alg} yields convergent results, we ran the algorithm for each $N \in \{10^3, 10^4, 10^5, 10^6\}$ for 50 trials. Each trial consisted of setting $(M, \mu, \lambda) = (8, 0.7, 1)$ and simulating $N$ i.i.d samples of a leaf $u \in \partial T$ a distance $t=t_u = 1$ away from the root.

See the three figures below. Each box-and-whisker plot approaches the horizontal blue line indicating the exact value of the variable as $N$ gets larger, indicating convergence.
\FloatBarrier 
\begin{figure} [H]
\centering
\begin{minipage}{.5\textwidth}
  \centering
  \includegraphics[scale=0.7]{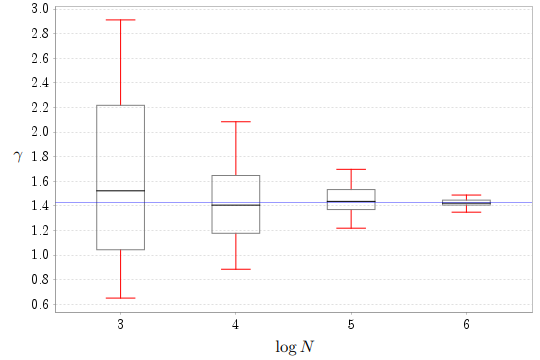}
\end{minipage}%
\begin{minipage}{.5\textwidth}
  \centering
  \includegraphics[scale=0.7]{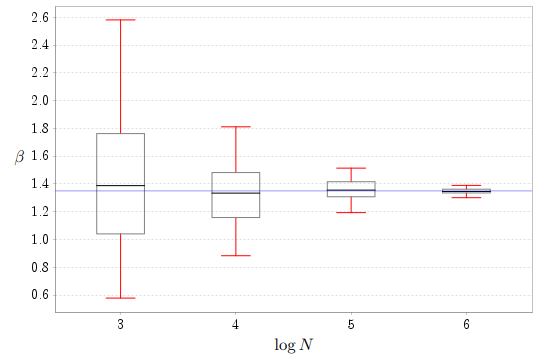}
\end{minipage}
\begin{minipage}{.5\textwidth}
  \centering
  \includegraphics[scale=0.7]{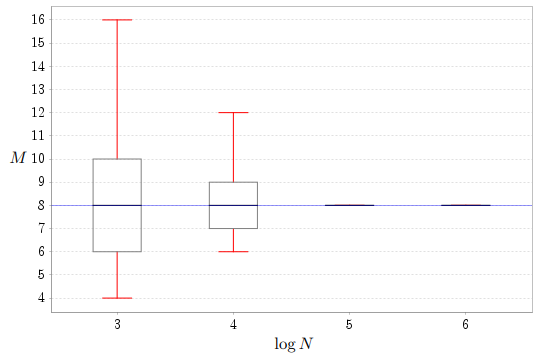}
\end{minipage}
\caption{Box-and-whisker plots for $\gamma, \beta,$ and $M$. Whiskers are drawn based on the largest or smallest data point within 1.5 times the interquartile range. The horizontal blue line in each figure shows the exact value of the variable.}
\end{figure}
\FloatBarrier

\subsection{Ancestral 1-mer and $\nu t_u$ reconstruction}\label{S:simulation_1mer}
In this section, we present an algorithm to reconstruct the ancestral 1-mer and $\nu t$, given $M$, $\lambda t$, $\pi_0$, as well as $\{(X_i, Z_i)\}_{i=1}^N$, the 1-mers of the $N$ sampled sequences. We straightforwardly apply the formulas found in Proposition \ref{prop:1merpi0}.

\FloatBarrier
\begin{algorithm}
\caption{1-mer and $\nu t_u$ reconstruction} \label{1-mer:alg}
\begin{algorithmic} 
\State \textbf{Input}: $N$, $M$, $\lambda t$, $\mu t$, $\pi_0,$ $\pi_1,$ $(X_1, Z_1), (X_2, Z_2), \dots, (X_N, Z_N)$ \\
\For {\textbf{each} $(i, j) \in \{(1, 0), (0, 1), (2, 0), (1, 1), (0, 2) \}$}
 \State
$ E_{i, j} \gets  \frac{1}{N}\sum_{n = 1}^N \frac{X_n\,!}{(X_n-i)!} \frac{Z_n\,!}{(Z_n-j)!}1_{X_n \ge i}1_{Z_n \ge j}.
$
\EndFor
\\
\State $H_{y_1} \gets \pi_0 E_{1,0} - \pi_1  E_{0,1}$
\State $A \gets -M\pi_1^2 + M\pi_1 (\pi_1^2 - \pi_0^2)$
\State $B \gets e^{\mu t}(\pi_1^2 - \pi_0^2)H_{y_1}$
\State $C \gets - e^{2\mu t}\left[\pi_0^2 E_{2,0} - 2\pi_0\pi_1 E_{1,1} + \pi_1^2  E_{0,2} - H_{y_1}^2\right]$
\\
\State $ e^{- \nu t}  \gets \frac { -B - \sqrt{B^2 - 4AC}}{2A}$\\
\State $ \nu t \gets -\ln (  e ^{-\nu t})$

\State $a \gets \frac{e^{\mu t} H_{y_1} + M \pi_1  e^{-\nu t}}{ e^{-\nu t}}.$
\\
\State \textbf{return} $a, \nu t.$
\end{algorithmic}
\end{algorithm}
\FloatBarrier

Again, we experimentally verified the algorithm by running it
for each $N \in \{10^3, 10^4, 10^5, 10^6\}$ for 50 trials. Each trial consisted of setting $(\mu, \lambda, \nu, \pi_0, \vec x) = (0.7, 1, 0.2, 0.3, 111100)$ and simulating $N$ i.i.d samples of a leaf $u$ a distance $t=t_u = 1$ away from the root.

The two figures below display the results. They show that even for small $N$, the median of the 50 trials is very accurate. 
\FloatBarrier 
\begin{figure} [H]
\centering
\begin{minipage}{.5\textwidth}
  \centering
  \includegraphics[scale=0.7]{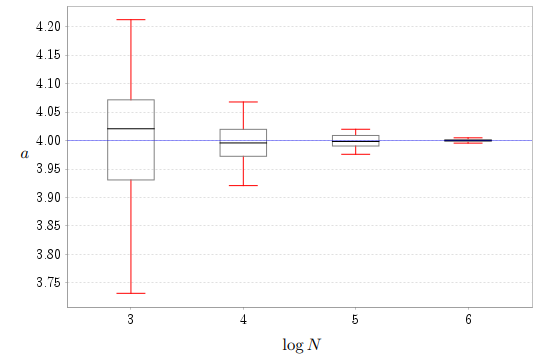}
\end{minipage}%
\begin{minipage}{.5\textwidth}
  \centering
  \includegraphics[scale=0.7]{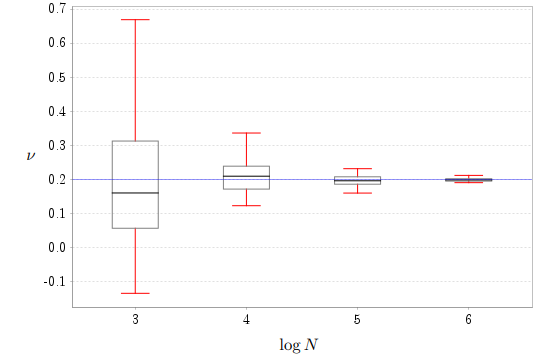}
\end{minipage}
\caption{Box-and-whisker plots for $a$ and $\nu.$ Whiskers are drawn based on the largest or smallest data point within 1.5 times the interquartile range. The horizontal blue line in each figure shows the exact value of the variable.}
\end{figure}
\FloatBarrier

\subsection{Root state reconstruction}\label{S:simulation_x}
For root state reconstruction, we are given $M=|\vec{x}|, \lambda t, \mu t, \nu t, \pi_0,$ and $\pi_1=1-\pi_0$. The goal is to reconstruct the sequence at the root of the tree using the $N$ sampled sequences $\{\mathcal I_{v}\}_{1 \le v \le N}$. Fix any $t < s_1 < s_2 < \dots < s_M < \infty$, which can in practice be picked by e.g. setting $s_j = t + t\cdot j.$ We only need to be able to calculate $(\lambda - \mu)(s_j-t)$ and $(\mu + \nu) (s_j - t)$, not $s_j$ itself, so this choice works.

Since $\mathcal I$ is a Markov process,
\[ p_{\vec x}^{\sigma}(s_j) = \mathbb{E}_{\vec x}[p_{\mathcal I_t}^{\sigma}(s_j - t)],\]
which gives us an estimator as follows. For each $1 \le v \le N,$ let $\mathcal{I}_v(i)$ be the $i$-th digit of $\mathcal{I}_v$. Then, we estimate
\begin{align*}
\widehat{p^{\sigma}_{\vec{x}}}(s_j) =&\, \frac{1}{|\partial T|}\sum_{v = 1}^N p^{\sigma}_{\mathcal{I}_v}(s_j - t) \\
=&\, \frac{1}{|\partial T|}\sum_{v = 1}^N\left(
\pi_{\sigma}\,\phi(s_j-t)\,\big[1-\big(\eta(s_j - t)\big)^{|\mathcal{I}_v|}\big]\,+\,\psi(s_j - t)\sum_{i=1}^{|\mathcal{I}_v|}1_{\{\mathcal{I}_v(i)=\sigma\}}\big(\eta(s_j - t)\big)^{i-1}
\right). 
\end{align*}
In the last equality we used Lemma~\ref{L:theta_v} in the Appendix.

From here, we could directly estimate $V, \Psi, U$ and then apply \eqref{E:Solve for x3} to estimate $\vec x.$ However, this leads to poor convergence because the entries of $V^{-1}$ and $\Psi^{-1}$ are large, resulting in inaccurate results unless $N$ is very large. For an algorithm with fast convergence, we can directly solve the matrix equation \eqref{E:Solve for x2} via an exhaustive search for the best $Y^{\vec x}$. In particular, letting $Y^{\vec x}_1$ and $U_1$ be the columns corresponding to $\sigma = 1$ of $Y^{\vec x}$ and $U$, respectively, we can estimate $Y^{\vec x}_1$ by computing
$$ Y^{\vec x}_1 \approx \arg \min _{ v \in \{0, 1\}^M} \|U_1 - \Psi V v \|_2.$$

Algorithm~\ref{rootstate:alg} uses this fast convergence approach. As in the previous sections, we experimentally verified the algorithm by running it
for each $N \in \{10^3, 10^4, 10^5, 10^6\}$ for 50 trials. Each trial consisted of setting $(\mu, \lambda, \nu, \pi_0, \vec x) = (0.4, 1, 0.2, 0.3, 11010111)$ and simulating $N$ i.i.d. samples $\{\mathcal I_v\}_{1 \le v \le n}$ of a leaf $u$ a distance $t = t_u = 1$ away from the root.

Figure~\ref{fig:difference} shows that the median trends downward to the best possible result of a difference of 0. Here, ``difference" refers to the number of different locations where the reconstructed root and the actual root 11010111 differ. Since this difference is a small integer bounded by $M = 8$, the median is limited to a small range, so the median itself is insufficient for showing the increased accuracy of the algorithm as $N$ increases. In fact, the median for both $N = 10^5$ and $N = 10^6$ is 0, indicating that at least half of the trials for those $N$ resulted in the exact root being reconstructed. To better show the increased accuracy as $N$ increases, we have included the mean in the box-and-whisker plot. We can see that even though the median does not change from $N = 10^3$ to $10^4$ and from $N = 10^5$ to $N = 10^6$, the mean is consistently decreasing. Moreover, at $N = 10^6$, the mean is around 0.25, indicating that the algorithm is very accurate.
\FloatBarrier
\begin{figure} [H]
\centering
\includegraphics[scale=1]{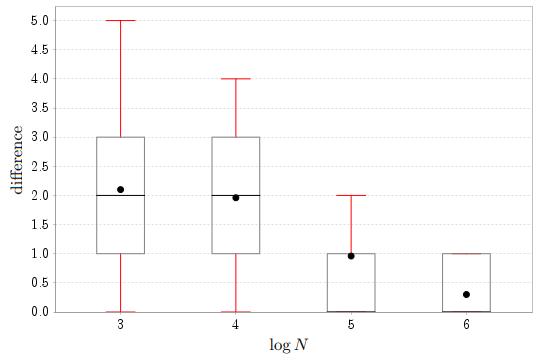}
\caption{Box-and-whisker plot for the number of differences between the reconstructed root and the actual root. Whiskers are drawn based on the largest data point within 1.5 times the interquartile range. The mean is represented as a dot.}
\label{fig:difference}
\end{figure}
\FloatBarrier
\begin{algorithm}
\caption{Root state reconstruction} \label{rootstate:alg}
\begin{algorithmic} 
\State \textbf{Input}: $N$, $M$, $\lambda t_u$, $\mu t_u$, $\nu t_u$, $\pi_0,$ $\pi_1,$ $\{\mathcal{I}_v\}_{1 \le v \le N}$\\
\For {$j = 1$ to $M$} \State 
$\eta_j \gets \eta(t_u + t_u\cdot j)$
\State $\phi_j \gets \phi(t_u+t_u \cdot j)$
\State $\psi_j \gets \psi(t_u + t_u\cdot j)$
\State $p^1_{\vec x}(t_u+t_u\cdot j) \gets \frac 1 N \sum_{v = 1}^N \left( \pi_1 \phi(j) [1 - (\eta(j))^{|\mathcal{I}_v|}] + \psi(j) \sum_{i = 1}^{|\mathcal{I}_v|} 1_{ \{ \mathcal{I}_v(i) = 1\}} (\eta(j))^{i-1} \right)$
\EndFor
\\

\State $\Psi \gets 0_{M \times M}$
\For {$j = 1$ to $M$} \State 
$\Psi_{jj} \gets \psi_j$
\EndFor
\\

\State $U_1 \gets 0_{M \times 1}$
\For {$j = 1$ to $M$} \State
$(U_1)_j = p_{\vec x}^1(t_u + t_u\cdot j) - \pi_1 \phi_j[1 - \eta_j^M] $
\EndFor
\\

\State $V \gets 0_{M \times M}$
\For {$i = 1$ to $M$ }
\For {$j = 1$ to $M$ } \State 
$V_{ij} = \eta_i^{j-1}$
\EndFor
\EndFor
\\

\State $Y_1^{\vec x} \gets \arg \min _{ v \in \{0, 1\}^M} \|U_1 - \Psi V v \|_2$

\State \textbf{return} $Y_1^{\vec x}$.
\end{algorithmic}
\end{algorithm}
\FloatBarrier

\subsection{Pairwise distance reconstruction}\label{S:simulation_T}
The setting in this section is different from the previous 3 subsections. Instead of simulating $N$ sequences and reconstructing some parameters of the tree, we will simulate $N$ samples of a fork tree (Figure \ref{fig:fork}) with edges $\{r, w\}, \{w, u\}, \{w, v\}$. Here, $r$ is the root node, and the goal is to estimate the distance $t_{uv}$ between $u$ and $v$, as well as the distance $t_w$ between $r$ and $w$. We are given $M, \lambda t_u, \lambda t_v, \mu t_u, \mu t_v$, where $t_u$ and $t_v$ are the distances between $r$ and the leaves $u$ and $v$, respectively. 
Our results in Section \ref{sec:pairwise} lead to the following algorithm. 

\FloatBarrier
\begin{algorithm}
\caption{Pairwise distance reconstruction} \label{pairwisedistance:alg}
\begin{algorithmic} 
\State \textbf{Input}: $N$, $M$, $\lambda t_u,$ $\lambda t_v$,  $\mu t_u$, $\mu t_v$, $N$ samples of joint distribution $(L^u, L^v)$ \\

\State $m^u \gets M e^{-(\mu t_u - \lambda t_u)}$
\State $m^v \gets M e^{-(\mu t_v - \lambda t_v)}$
\State $Cov_M(L^u, L^v) \gets \frac 1 N \sum_{(L^u, L^v)} (L^u - m^u)(L^v - m^v)$
\\

\State $e^{\frac{-(\mu-\lambda)}{2}\,t_{uv}} \gets \frac{\mu - \lambda}{M(\lambda +\mu)}\,e^{\frac{+(\mu-\lambda)}{2}\,(t_u+t_v)}\,Cov_M(L^u,L^v)\;+\; e^{\frac{-(\mu-\lambda)}{2}\,(t_u+t_v)}$
\State $ (\mu - \lambda) t_{uv} \gets -2 \ln e^{\frac{-(\mu-\lambda)}{2}\,t_{uv}}$
\State $\mu t_{uv} \gets (\mu - \lambda) t_{uv} \frac { \mu t_u }{(\mu - \lambda) t_u}$\\

\State $\mu t_w \gets \,\frac{\mu t_u+\mu t_v-\mu t_{uv}}{2}$
\\

\State \textbf{return} $\mu t_{uv}, \mu t_w$.
\end{algorithmic}
\end{algorithm}
\FloatBarrier
We ran Algorithm \ref{pairwisedistance:alg} for each $N \in \{10^3, 10^4, 10^5, 10^6\}$ for 50 trials. Each trial consisted of creating $N$ fork trees with edges $\{r, w\}, \{w, u\}, \{w, u\}$, with parameters $t_w = 1, t_u = 3, t_v = 4$, and with $(\mu, \lambda, M) = (0.3, 0.5, 8).$

Figure \ref{fig:pariwisedistance} shows that even at small $N$, the median is very close to the exact value. As $N$ grows larger, the boxes and whiskers get smaller, converging toward the correct value.

\begin{figure} [H]
\centering
\begin{minipage}{.5\textwidth}
  \centering
  \includegraphics[scale=0.7]{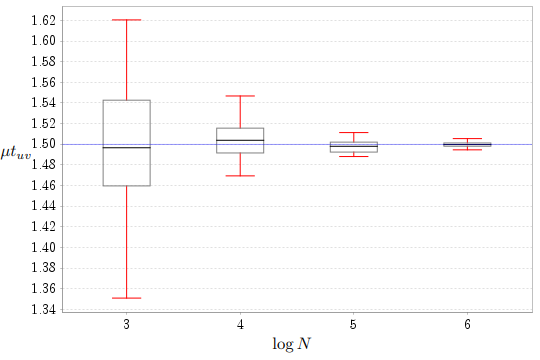}
\end{minipage}%
\begin{minipage}{.5\textwidth}
  \centering
  \includegraphics[scale=0.7]{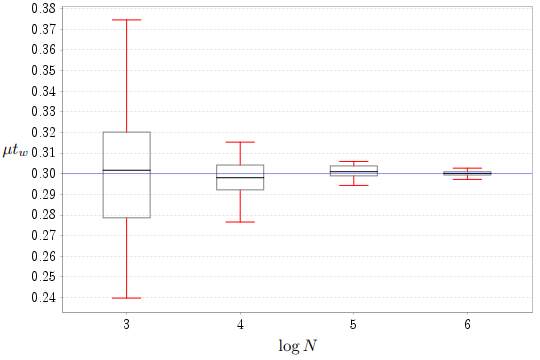}
\end{minipage}
\caption{Box-and-whisker plots for $\mu t_{uv}$ and $\mu t_w.$ Whiskers are drawn based on the largest or smallest data point within 1.5 times the interquartile range. The horizontal blue line in each figure shows the exact value of the variable.}
\label{fig:pariwisedistance}
\end{figure}
\FloatBarrier
\newpage

\section{Acknowledgments}

This work started as a Research Experience for Undergraduate (REU) project of A. Xue at Indiana University under National Science Foundation grant DMS-2051032. 
For B. Legried, this research is supported by the NSF-Simons Southeast Center for Mathematics and Biology (SCMB) through the grants National Science Foundation DMS-1764406 and Simons Foundation/SFARI 594594, and by the Department of Statistics RTG at the University of Michigan under DMS-1646108. 
This work is also supported by National Science Foundation grants DMS-2152103 and DMS-2348164 to W.-T. L. Fan.

\section{Appendix}

\subsection{Sequence length Process}

The sequence length of the TKF91 edge process is a continuous-time linear birth-death-immi\-gration process $(|\mathcal{I}_{t}|)_{t\geq 0}$ with infinitesimal generator $Q_{i,i+1}=\lambda_0+i\lambda$ (for $i\in \ZZ_+$), $Q_{i,i-1}=i\mu$ (for $i\geq 1$) and $Q_{i,j}=0$ otherwise. Note that $\lambda_0=0$ in  Definition \ref{Def_BiINDEL}.

This is a well-studied process for which explicit forms for the transition density $p_{ij}(t)$ and probability generating functions $G_i(z,t)=\sum_{j=0}^{\infty}p_{ij}(t)z^j$ are known. This process was also analyzed in~\cite{THATTE200658} in the related context of phylogeny estimation.
We collect here a few properties that will be useful in our analysis.
The probability generating function is given by
\begin{equation*}
G_i(z,t)=\left[\frac{1-\beta -z(\gamma-\beta)}{1-\beta\gamma -\gamma z(1-\beta)}\right]^i\,\left[\frac{1-\gamma}{1-\beta\gamma -\gamma z(1-\beta)}\right]^\delta
\end{equation*}
for  $i\in \ZZ_+$ and $t>0$, where 
\begin{equation*}
\beta =\beta_t = e^{-(\mu-\lambda)t}, \quad \gamma =\frac{\lambda}{\mu}
\quad \text{and}\quad \delta=\frac{\lambda_0}{\lambda}.
\end{equation*}

\medskip

When $\lambda_0=0$ (i.e. $\delta=0$), we obtain the PGF considered in Lemma \eqref{L:InvertLength}. The expected value and the second moment are given by 
\begin{align*}
\E_{i}[|\mathcal{I}_{t}|]&=
i\beta \quad\text{ and} 
\\
\E_{i}[|\mathcal{I}_{t}|^2]&= i^2\beta^2+ i\frac{\beta(2 \gamma - \beta \gamma - \beta)}{1 - \gamma}.   
\end{align*}
Appendix of \cite{fan2020statistically} contains properties of  $\mathcal{I}$ when $\lambda_0=\lambda$. 

\subsection{The 1-mer process}
\label{s:1-merAppendix}
To obtain the formula for the probability generating function (PFG) of 1-mer, we start by differentiating the PGF and applying the method of characteristics.
\begin{lemma}\label{L:1merPDE}
Fix $(a,b)\in \ZZ_+^2$ and let $p_t(j,k):=\P_{(a,b)}(X_t=j, Z_t=k)$.
Then the function $G(z_1,z_2;\,t):=G_{(a,b)}(z_1,z_2;\,t)$ satisfies the differential equation
\begin{align}
\frac{\partial G}{\partial t} &= \left(z_1^2 \lambda \pi_1 + z_1z_2 \lambda \pi_0 + \mu + z_2 \nu \pi_0  - z_1(\lambda + \mu + \nu\pi_0)\right) \,\frac{\partial G}{\partial z_1}  \label{1merPDE1}
   \\&
   \quad + \left( z_1z_2 \lambda \pi_1 + z_2^2 \lambda \pi_0 + \mu + z_1\nu\pi_1 - z_2(\lambda + \mu + \nu\pi_1)\right)\,\frac{\partial G}{\partial z_2}, \label{1merPDE2}
\end{align}
with initial condition $G(z_1, z_2; 0) = z_1^a z_2^b$ and boundary conditions 
\begin{equation}\label{1mer_BC1}
G(0, 0; t) = p_t(0,0), \quad G(1, 1; t) = 1, 
\end{equation}
and
\begin{align} \label{1mer_BC2}
G(z_1, 0; t) &= \mathbb{E}_{(a,b)}[z_1^{X_t}|Z_t = 0] \mathbb{P}_{(a,b)}(Z_t = 0), \\
G(0, z_2; t) &= \mathbb{E}_{(a,b)}[z_2^{Z_t}| X_t = 0] \mathbb{P}_{(a,b)}(X_t = 0).
\end{align}

\end{lemma}

\begin{proof}
Using the transition rates \eqref{imer_rates}, we find that for each fixed $(a,b)\in \mathbb{Z}_+^2$, the transition probabilities $p_t(j,k)$
satisfy the system of differential equations
\begin{align*}
\frac{\partial p_t(j,k)}{\partial t} &= p_t(j-1, k) \lambda (j + k - 1)\pi_1 + 
p_t(j, k-1) \lambda(j+k-1) \pi_0 \\&\quad + 
p_t(j+1, k) \mu(j+1) + 
p_t(j, k+1) \mu(k+1) \\&\quad+ 
p_t(j+1, k-1) \nu (j + 1) \pi_0 + 
p_t(j-1, k+1) \nu(k+1)\pi_1 \\&\quad- 
p_t(j,k) \left(\lambda(j+k) + \mu(j+k) + \nu j\pi_0 + \nu k \pi_1 \right), \qquad\qquad t\in(0,\infty),
\end{align*}
for all $(j,k)\in\ZZ_+^2$.

Note that $G(z_1,z_2;\,t)= \sum_{(j,k)\in\ZZ_+^2} p_t(j, k) z_1^{j}z_2^{k}\,$ satisfies
$$\frac{\partial G}{\partial z_1} = \sum p_t(j, k)j z_1^{j-1}z_2^{k} \quad\text{ and }
\quad \frac{\partial G}{\partial z_2} = \sum p_t(j, k)k z_1^{j}z_2^{k-1}.$$
Thus,
\begin{align*}
    \frac{\partial G}{\partial t} &= \sum \frac{\partial p_t(j,k)}{\partial t} z_1^{j}z_2^{k} \\
    &= \sum z_1^{j}z_2^{k}p_t(j-1, k) \lambda (j + k - 1)\pi_1 + 
   \sum  z_1^{j}z_2^{k}p_t(j, k-1) \lambda(j+k-1) \pi_0 \\&\quad + 
    \sum z_1^{j}z_2^{k}p_t(j+1, k) \mu(j+1) + 
   \sum  z_1^{j}z_2^{k}p_t(j, k+1) \mu(k+1) \\&\quad+ 
    \sum z_1^{j}z_2^{k}p_t(j+1, k-1) \nu (j + 1) \pi_0 + 
    \sum z_1^{j}z_2^{k}p_t(j-1, k+1) \nu(k+1)\pi_1 \\&\quad- 
   \sum  z_1^{j}z_2^{k}p_t(j,k) \left(\lambda(j+k) + \mu(j+k) + \nu j\pi_0 + \nu k \pi_1 \right) \\
    &= \sum z_1^{j+1}z_2^{k}p_t(j, k) \lambda (j + k)\pi_1 + 
   \sum  z_1^{j}z_2^{k+1}p_t(j, k) \lambda(j+k) \pi_0 \\&\quad + 
    \sum z_1^{j-1}z_2^{k}p_t(j, k) \mu j + 
   \sum  z_1^{j}z_2^{k-1}p_t(j, k) \mu k \\&\quad+ 
    \sum z_1^{j - 1}z_2^{k + 1}p_t(j, k) \nu j \pi_0 + 
    \sum z_1^{j + 1}z_2^{k - 1}p_t(j, k) \nu k\pi_1 \\&\quad- 
   \sum  z_1^{j}z_2^{k}p_t(j,k) \left(\lambda(j+k) + \mu(j+k) + \nu j\pi_0 + \nu k \pi_1 \right) \\
   &= \sum p_t(j, k) j \left[ z_1^{j + 1}z_2^{k} \lambda \pi_1 + z_1^{j} z_2^{k + 1} \lambda \pi_0 + z_1^{j - 1}z_2^{k} \mu + z_1^{j - 1}z_2^{k + 1} \nu \pi_0 - z_1^{j}z_2^{k} (\lambda + \mu + \nu \pi_0 ) \right] \\ 
   &\quad +\sum p_t(j, k) k \left[ z_1^{j + 1}z_2^{k} \lambda \pi_1 + z_1^{j} z_2^{k + 1} \lambda \pi_0 + z_1^{j}z_2^{k-1} \mu + z_1^{j + 1}z_2^{k - 1} \nu \pi_1 - z_1^{j}z_2^{k} (\lambda + \mu + \nu \pi_1 ) \right] \\
   &=\left(z_1^2 \lambda \pi_1 + z_1z_2 \lambda \pi_0 + \mu + z_2 \nu \pi_0  - z_1(\lambda + \mu + \nu\pi_0)\right) \,\frac{\partial G}{\partial z_1} 
   \\&
   \quad + \left( z_1z_2 \lambda \pi_1 + z_2^2 \lambda \pi_0 + \mu + z_1\nu\pi_1 - z_2(\lambda + \mu + \nu\pi_1)\right)\,\frac{\partial G}{\partial z_2}.
\end{align*}
\end{proof}

\begin{proof}
[Proof of Proposition \ref{prop:1mer}]
We now solve the Lagrangian-type PDE in Lemma \ref{L:1merPDE} by the method of characteristics.
Let $\Gamma$ be the initial condition curve $\{(r_1, r_2; 0) \mid G(r_1, r_1; 0) \text{ exists}\}$. From above, we have $G(r_1, r_2; 0) = r_1^a r_2^b$ for all points on $\Gamma.$

Let $u(\vec r, s) = u(\vec r(z_1, z_2;t), s(z_1, z_2;t)) = G(z_1, z_2; t).$ We have the following four characteristic equations:
\begin{align}
    \frac{dt(\vec r, s)}{ds}&= 1  \label{dz1} \\ 
    \frac{dz_1(\vec r, s)}{ds}&= -(z_1^2 \lambda \pi_1 + z_1z_2 \lambda \pi_0 + \mu + z_2 \nu \pi_0  - z_1(\lambda + \mu + \nu\pi_0)) \label{dz2} \\ 
    \frac{dz_2(\vec r, s)}{ds}&= -(z_1z_2 \lambda \pi_1 + z_2^2 \lambda \pi_0 + \mu + z_1\nu\pi_1 - z_2(\lambda + \mu + \nu\pi_1)) \notag\\
    \frac{d u(\vec r, s)}{ds}&=  0, \notag
\end{align}
with boundary conditions
\begin{align*}
    t(\vec r, 0) &= 0 \\
    z_1(\vec r, 0) &= r_1 \\
    z_2(\vec r, 0) &= r_2\\
    u(\vec r, 0) &= r_1^a r_2^b.
\end{align*}
We want to express $r_1$ and $r_2$ in terms of $z_1, z_2,$ and $t$. The equations \eqref{dz1} and \eqref{dz2} are difficult to solve directly, so we substitute  $$(y_1,y_2)=(z_1-z_2,\,\pi_1z_1+\pi_0z_2).$$
We have
\begin{align*}
    \frac{dy_1(\vec r, s)}{ds} &= y_1 \left(\lambda +\mu +\nu -\lambda  y_2\right) \\
    \frac{dy_2(\vec r, s)}{ds} &= \left(y_2-1\right) \left(\mu -\lambda  y_2\right).
\end{align*}
Hence, in the case where $\lambda \ne \mu$, we have
\begin{align*} 
    y_1(\vec r, s) &= \frac{c_1 (\lambda -\mu ) e^{s (\lambda +\nu )}}{-c_2 \lambda +c_2 \lambda  e^{s (\lambda -\mu )}+\lambda -\mu  e^{s (\lambda -\mu )}} \\
    y_2(\vec r, s) &= \frac{e^{s\lambda } (c_2 \lambda -\mu )-(c_2-1) \mu  e^{s\mu}}{e^{s\lambda} (c_2 \lambda -\mu )-(c_2-1) \lambda  e^{s\mu }},
\end{align*}
where $c_1 = y_1(\vec r, 0) = r_1 - r_2,$ and $c_2 = y_2(\vec r, 0) = \pi_1r_1 + \pi_0r_2.$
Solving for $r_1$ and $r_2$ in terms of $y_1, y_2,$ and $t$, and then writing in terms of $z_1, z_2,$ and $t$, we find that
\begin{align*} 
&G(z_1, z_2; t) = u(\vec r(z_1, z_2; t), s(z_1,z_2;t)) = u(\vec r(z_1, z_2;t), 0) = r_1^a r_2^b \\ &= e^{-a \nu  t} \left(\left(\mu -\lambda  y_2\right)e^{\mu  t}+\lambda  \left(y_2-1\right) e^{\lambda 
   t}\right){}^{-a-b}
   \\&\quad  \left(-\mu  e^{t (\lambda +\nu )}+\mu  e^{t (\mu +\nu )}+y_2 \left(\mu 
   e^{t (\lambda +\nu )}-\lambda  e^{t (\mu +\nu )}\right)+\pi _0 y_1 (\mu -\lambda )\right){}^a \\&\quad \left(-\mu e^{\lambda  t}+\mu  e^{\mu  t}+\pi_1 y_1 (\lambda -\mu ) e^{-\nu t}+\mu
    y_2 e^{\lambda  t}-\lambda  y_2 e^{\mu  t}\right){}^b
\\&= 
\left(e^{\mu  t} \left(\mu -\pi_1 \lambda  z_1-\pi _0 \lambda 
   z_2\right)+\lambda  \left(\pi_1 z_1+\pi _0 z_2-1\right) e^{\lambda  t}\right)^{-a-b}
   \\& \quad 
  \left(\mu  \left(e^{\mu t}-e^{\lambda t }\right)+z_1 \left(\pi _0 (\mu -\lambda
   ) e^{-\nu t} -\pi_1 \left(\lambda  e^{\mu t}-\mu  e^{\lambda t
   }\right)\right)+\pi _0 z_2 \left(\lambda \left (e^{-\nu t} - e^{\mu t}\right)+\mu  \left(e^{ \lambda t}- e^{-\nu t}\right)\right)\right)^a 
   \\& \quad
   \left(\mu (-e^{\lambda  t} +  e^{\mu  t})+\pi_1 \left(z_1-z_2\right) (\lambda -\mu ) e^{-\nu t}+\mu  \left(\pi _0 z_2+ \pi_1 z_1\right) e^{\lambda  t}-\lambda  \left(\pi_1 z_1+\pi _0 z_2\right) e^{\mu
   t}\right)^b. 
 \end{align*}
The proof is complete upon the  substitutions
$y_1 = z_1 - z_2$ and $y_2 = \pi_1 z_1 + \pi_0 z_2$.

Now assume that $\mu = \lambda > 0$.  Then \begin{equation*}
    \frac{dy_2}{ds} = -\lambda(1-y_2)^2,
\end{equation*} where the general solution is \begin{equation*}
    y_2(\vec{r},s) = \frac{\lambda s - c_2(1 + \lambda s)}{-1 + \lambda s - c_2\lambda s}.
\end{equation*} 
Making this substitution implies \begin{equation*}
    \frac{dy_1}{ds} = y_1\left(2\lambda + \nu - \lambda \frac{\lambda s - c_2(1 + \lambda s)}{-1 + \lambda s - c_2\lambda s}\right).
\end{equation*} The solution is then \begin{equation*}
    y_1(\vec{r},s) = \frac{c_1 e^{(\lambda + \nu)s}}{1+\lambda s(c_2 - 1)}.
\end{equation*}

\end{proof}

\subsubsection{System of equations for $1$-mer}
\label{sss:1-mer System}

In this section, we derive the system of equations involving the triple $(\pi_0, a,\nu t)$ in view of Proposition \ref{prop:1mer}.  We provide this system for completeness,  but we do not need this section for the proof of any stated result in this paper. 

The key point is that it is not clear whether this system has a unique solution among the possible parameters, unless we assume that one of  $\pi_0$, $a$ and $\nu t$ is known. We post it as an open problem to solve for the triple $(\pi_0, a,\nu t)$.

Consider
$$C_{(i,j)}:= \frac{\partial^{(i+j)} \ln G_{(a,b)}(z_1,z_2;t)}{\partial^{(i)} z_1 \,\partial^{(j)} z_2} \Big|_{z_1 = 1, z_2=1}.$$
We will compute the first three partial derivatives in $z_1$, but this requires the multivariable chain rule applied to the underlying variables $y_1$ and $y_2$.  We have \begin{equation*}
    \partial_{z_1} \ln G = \partial_{y_1}\ln H + \pi_1\partial_{y_2} \ln H.
\end{equation*} Consider $$\ln H = \ln(r_1(y_1,y_2)^a r_2(y_1,y_2)^b) = a \ln(r_1(y_1,y_2)) + b \ln(r_2(y_1,y_2)).$$ Since $$\frac{\partial}{\partial y_1}r_{i+1} = \frac{(-1)^{i}\pi_i (\mu - \lambda)e^{-\nu t}}{(\mu - \lambda y_2)e^{\mu t} + \lambda(y_2 - 1)e^{\lambda t}} = \frac{(-1)^{i}C \pi_i e^{-\nu t} B }{\gamma A - \gamma y_2},$$ we have \begin{align*}
    \frac{\partial}{\partial y_1}\ln H = a \frac{1}{r_1} \frac{\partial r_1}{\partial y_1} + b \frac{1}{r_2}\frac{\partial r_2}{\partial y_1} &= \frac{a\pi_0(\mu - \lambda)e^{-\nu t}}{-\mu  e^{\lambda t}+\mu  e^{\mu t}+y_2 \left(\mu 
   e^{\lambda t}-\lambda  e^{\mu t}\right)+\pi _0 y_1 (\mu -\lambda )e^{-\nu t}} \\
   &- \frac{b \pi_1 (\mu - \lambda)e^{-\nu t}}{-\mu e^{\lambda  t}+\mu  e^{\mu  t}+ y_2(\mu e^{\lambda  t}-\lambda   e^{\mu  t}) +\pi_1 y_1 (\lambda -\mu ) e^{-\nu t}} \\
   &= a \frac{\pi_0 e^{-\nu t}B}{1 - \gamma A y_2 + \pi_0 e^{-\nu t}B y_1} - b\frac{\pi_1 e^{-\nu t}B}{1 - \gamma A y_2 - \pi_1 e^{-\nu t}B y_1} .
\end{align*} Similarly, for each $i \in \{0,1\}$ we have \begin{align*}
    \frac{\partial}{\partial y_2} r_{i+1} &= \bigg[ (\mu e^{\lambda t} - \lambda e^{\mu t})\left( (\mu - \lambda y_2)e^{\mu t} + \lambda(y_2 - 1)e^{\lambda t}\right) \\
    &- \left(-\mu e^{\lambda t} + \mu e^{\mu t} + y_2(\mu e^{\lambda t} - \lambda e^{\mu t}) + (-1)^{i}\pi_i y_1 (\mu - \lambda)e^{-\nu t}\right)(-\lambda e^{\mu t} + \lambda e^{\lambda t})\bigg] \\
    &/ \left[(\mu - \lambda y_2)e^{\mu t} + \lambda(y_2 - 1)e^{\lambda t}\right]^2 \\
    &= \frac{(\mu e^{\lambda t} - \lambda e^{\mu t})(\mu e^{\mu t} - \lambda e^{\lambda t}) - \lambda (e^{\lambda t} - e^{\mu t})(\mu e^{\mu t} - \mu e^{\lambda t} + (-1)^{i}\pi_i y_1 (\mu - \lambda) e^{-\nu t})}{\left[(\mu - \lambda y_2)e^{\mu t} + \lambda(y_2 - 1)e^{\lambda t}\right]^2} \\
    &= \frac{-C\gamma\left(1 + \gamma A + (-1)^{i}\pi_i e^{-\nu t} B y_1\right)}{(\gamma A - \gamma y_2)^2}.
\end{align*} Consider $$F_{\pi_i,\nu t}(y_1) = -\gamma(1 + \gamma A + (-1)^{i}\pi_i e^{-\nu t}B y_1)$$ and \begin{align*}s_{i+1} &= \left[(\mu - \lambda y_2)e^{\mu t} + \lambda(y_2 - 1)e^{\lambda t}\right]\left[-\mu  e^{\lambda t}+\mu  e^{\mu t}+y_2 \left(\mu 
   e^{\lambda t}-\lambda  e^{\mu t}\right)+(-1)^{i}\pi_i y_1 (\mu -\lambda )e^{-\nu t}\right] \\
   &= (\gamma A - \gamma y_2)(1 - \gamma Ay_2 + (-1)^{i}\pi_i e^{-\nu t}B y_1)\end{align*} for each $i \in \{0,1\}$, with derivatives \begin{align*}
       \frac{\partial}{\partial y_2}s_{i+1} &= (-\lambda e^{\mu t} + \lambda e^{\lambda t})\left[-\mu  e^{\lambda t}+\mu  e^{\mu t}+y_2 \left(\mu 
   e^{\lambda t}-\lambda  e^{\mu t}\right)+(-1)^{i}\pi_i y_1 (\mu -\lambda )e^{-\nu t}\right] \\
   &+\left[(\mu - \lambda y_2)e^{\mu t} + \lambda(y_2 - 1)e^{\lambda t}\right] \left(\mu e^{\lambda t} - \lambda e^{\mu t}\right) \\
   &= -\gamma(1-A) - \gamma (-1)^{i}\pi_i e^{-\nu t}By_1 + 2\gamma^2 A y_2 \\
   \frac{\partial^2}{\partial y_2^2}s_{i+1} &= (-\lambda e^{\mu t} + \lambda e^{\lambda t})(\mu e^{\lambda t} - \lambda e^{\mu t}) + (-\lambda e^{\mu t} + \lambda e^{\lambda t})(\mu e^{\lambda t} - \lambda e^{\mu t}) \\
   &= 2 \gamma^2 A.
   \end{align*} The first derivative in $y_2$ is \begin{align*}\frac{\partial}{\partial y_2} \ln H &= a \frac{1}{r_1} \frac{\partial r_1}{\partial y_2} + b \frac{1}{r_2} \frac{\partial r_2}{\partial y_2} = a \frac{F_{\pi_0,\nu t}(y_1)}{s_1} + b \frac{F_{\pi_1,\nu t}(y_1)}{s_2} \\
   &= a \frac{-\gamma(1+\gamma A + \pi_0 e^{-\nu t}B y_1)}{(\gamma A - \gamma y_2)(1 - \gamma A y_2 + \pi_0 e^{-\nu t}B y_1)} + b \frac{-\gamma(1+\gamma A -\pi_1 e^{-\nu t}B y_1)}{(\gamma A - \gamma y_2)(1 - \gamma A y_2 -\pi_1 e^{-\nu t}B y_1)}.\end{align*} From the multivariable chain rule, the second- and third-order partial derivatives of $\ln G$ are \begin{align*}
       \partial^{2}_{z_1} \ln G &= \partial^2_{y_1} \ln H + 2\pi_1 \partial^2_{y_1 y_2}\ln H + \pi_1^2 \partial^2_{y_2}\ln H \\
       \partial^{3}_{z_1} \ln G &= \partial^3_{y_1} \ln H + 3\pi_1 \partial^3_{y_1^2 y_2} + 3\pi_1^2 \partial^3_{y_1 y_2^2} \ln H + \pi_1^3 \partial^3_{y_2} \ln H.
   \end{align*} For all higher-order derivatives involving $y_1$ in whole or in part, it will be easier to differentiate $\partial_{y_1} \ln H$.  However, the higher-order derivatives involving only $y_2$ must be computed using $\partial_{y_2} \ln H$, which is a little more complicated.

   Observe in $\partial_{y_1} \ln H$ that the variables $y_1$ and $y_2$ appear only in denominators.  We have \begin{align*}
    \frac{\partial}{\partial y_1}\left(\frac{1}{r_{i+1}}\frac{\partial r_{i+1}}{\partial y_1}\right) &= -\frac{\left((-1)^{i}\pi_i(\mu - \lambda)e^{-\nu t}\right)^2}{\left(-\mu  e^{\lambda t}+\mu  e^{\mu t}+y_2 \left(\mu 
   e^{\lambda t}-\lambda  e^{\mu t}\right)+(-1)^{i}\pi _i y_1 (\mu -\lambda )e^{-\nu t}\right)^2} \\
   &= - \left(\frac{\pi_i e^{-\nu t} B}{1 - \gamma A y_2 + (-1)^{i}\pi_i e^{-\nu t}B y_1}\right)^2.
\end{align*} Then the second order partial derivative with respect to $y_1$ is \begin{align*}
    \frac{\partial^2}{\partial y_1^2} \ln H &= -a \left(\frac{\pi_0(\mu - \lambda)e^{-\nu t}}{-\mu  e^{\lambda t}+\mu  e^{\mu t}+y_2 \left(\mu 
   e^{\lambda t}-\lambda  e^{\mu t}\right)+\pi _0 y_1 (\mu -\lambda )e^{-\nu t}}\right)^2 \\
   &- b\left(\frac{\pi_1(\mu - \lambda)e^{-\nu t}}{-\mu  e^{\lambda t}+\mu  e^{\mu t}+y_2 \left(\mu 
   e^{\lambda t}-\lambda  e^{\mu t}\right)-\pi _1 y_1 (\mu -\lambda )e^{-\nu t}}\right)^2 \\
   &= -a \left(\frac{\pi_0 e^{-\nu t} B}{1 - \gamma A y_2 + \pi_0 e^{-\nu t}B y_1}\right)^2 - b \left(\frac{\pi_1 e^{-\nu t} B}{1 - \gamma A y_2 -\pi_1 e^{-\nu t}B y_1}\right)^2.
\end{align*} The third-order partial derivative is\begin{align*}
    \frac{\partial^3}{\partial y_1^3} \ln H &= 2a \left(\frac{\pi_0(\mu - \lambda)e^{-\nu t}}{-\mu  e^{\lambda t}+\mu  e^{\mu t}+y_2 \left(\mu 
   e^{\lambda t}-\lambda  e^{\mu t}\right)+\pi _0 y_1 (\mu -\lambda )e^{-\nu t}}\right)^3 \\
   &-2b\left(\frac{\pi_1(\mu - \lambda)e^{-\nu t}}{-\mu  e^{\lambda t}+\mu  e^{\mu t}+y_2 \left(\mu 
   e^{\lambda t}-\lambda  e^{\mu t}\right)-\pi _1 y_1 (\mu -\lambda )e^{-\nu t}}\right)^3 \\
   &= 2a \left(\frac{\pi_0 e^{-\nu t} B}{1 - \gamma A y_2 + \pi_0 e^{-\nu t}B y_1}\right)^3 - 2b \left(\frac{\pi_1 e^{-\nu t} B}{1 - \gamma A y_2 + \pi_1 e^{-\nu t}B y_1}\right)^3.
\end{align*} For the second-order cross-partial, we have \begin{align*}\frac{\partial}{\partial y_2}\left(\frac{1}{r_{i+1}} \frac{\partial r_{i+1}}{\partial y_1}\right) &= -\frac{(-1)^{i} \pi_i e^{-\nu t} (\mu - \lambda)(\mu e^{\lambda t} - \lambda e^{\mu t})}{\left[-\mu e^{\lambda t} + \mu e^{\mu t} + y_2(\mu e^{\lambda t} - \lambda e^{\mu t}) + (-1)^{i}\pi_i y_1 (\mu - \lambda)e^{-\nu t}\right]^2} \\
&= \gamma A \frac{(-1)^{i} \pi_i e^{-\nu t} B}{(1 - \gamma A y_2 + (-1)^{i}\pi_i e^{-\nu t}B y_1)^2},\end{align*} so that \begin{align*}
    \frac{\partial^2}{\partial y_1 y_2} \ln H &= -a \frac{ \pi_0 (\mu - \lambda)e^{-\nu t}(\mu e^{\lambda t} - \lambda e^{\mu t})}{\left[-\mu e^{\lambda t} + \mu e^{\mu t} + y_2(\mu e^{\lambda t} - \lambda e^{\mu t}) + \pi_0 y_1 (\mu - \lambda)e^{-\nu t}\right]^2} \\
    &+b \frac{ \pi_1 (\mu - \lambda)e^{-\nu t}(\mu e^{\lambda t} - \lambda e^{\mu t})}{\left[-\mu e^{\lambda t} + \mu e^{\mu t} + y_2(\mu e^{\lambda t} - \lambda e^{\mu t}) -\pi_1 y_1 (\mu - \lambda)e^{-\nu t}\right]^2} \\
    &= a \gamma A \frac{\pi_0 e^{-\nu t}B}{(1 - \gamma A y_2 + \pi_0 e^{-\nu t}B y_1)^2} - b \gamma A \frac{\pi_1 e^{-\nu t}B}{(1 - \gamma A y_2 - \pi_1 e^{-\nu t}B y_1)^2} .
\end{align*} Similarly, for the third-order cross-partials involving $y_1$, we have\begin{align*}
    \frac{\partial^3}{\partial y_1^2 y_2} \ln H &= 2a \frac{\left(\pi_0 (\mu - \lambda) e^{-\nu t}\right)^2 (\mu e^{\lambda t} - \lambda e^{\mu t})}{\left[-\mu e^{\lambda t} + \mu e^{\mu t} + y_2(\mu e^{\lambda t} - \lambda e^{\mu t}) + \pi_0 y_1 (\mu - \lambda)e^{-\nu t}\right]^3} \\
    &+2b \frac{\left(\pi_1 (\mu - \lambda) e^{-\nu t}\right)^2 (\mu e^{\lambda t} - \lambda e^{\mu t})}{\left[-\mu e^{\lambda t} + \mu e^{\mu t} + y_2(\mu e^{\lambda t} - \lambda e^{\mu t}) - \pi_1 y_1 (\mu - \lambda)e^{-\nu t}\right]^3} \\
    &= -2a \gamma A \frac{\left(\pi_0 e^{-\nu t} B \right)^2}{\left(1 - \gamma A y_2 + \pi_0 e^{-\nu t}B y_1\right)^3} -2b \gamma A \frac{\left(\pi_1 e^{-\nu t} B \right)^2}{\left(1 - \gamma A y_2 - \pi_1 e^{-\nu t}B y_1\right)^3}
\end{align*} and \begin{align*}
    \frac{\partial^3}{\partial y_1 y_2^2} \ln H &= 2a \frac{ \pi_0 (\mu - \lambda) e^{-\nu t}(\mu e^{\lambda t} - \lambda e^{\mu t})^2}{\left[-\mu e^{\lambda t} + \mu e^{\mu t} + y_2(\mu e^{\lambda t} - \lambda e^{\mu t}) + \pi_0 y_1 (\mu - \lambda)e^{-\nu t}\right]^3} \\
    &-2b \frac{ \pi_1 (\mu - \lambda)e^{-\nu t}(\mu e^{\lambda t} - \lambda e^{\mu t})^2}{\left[-\mu e^{\lambda t} + \mu e^{\mu t} + y_2(\mu e^{\lambda t} - \lambda e^{\mu t}) - \pi_1 y_1 (\mu - \lambda)e^{-\nu t}\right]^3} \\
    &= 2a(\gamma A)^2 \frac{\pi_0 e^{-\nu t} B}{\left(1 - \gamma A y_2 + \pi_0 e^{-\nu t} B y_1\right)^3} -  2b(\gamma A)^2 \frac{\pi_1 e^{-\nu t} B}{\left(1 - \gamma A y_2 - \pi_1 e^{-\nu t} B y_1\right)^3}.
\end{align*}

To further differentiate $\partial y_2 \ln H$, we note that $F_{\pi_i,\nu t}(y_1)$ has no dependence on $y_2$.  It is required to compute the first and second derivatives in $y_2$ for $s_{i+1}, i \in \{0,1\}$.  We then have \begin{align*}
    \frac{\partial^2}{\partial y_2^2} \ln H &= -a \frac{F_{\pi_0,\nu t}(y_1)}{s_1^2} \frac{\partial s_1}{\partial y_2} -b \frac{F_{\pi_1,\nu t}(y_1)}{s_2^2} \frac{\partial s_2}{\partial y_2} \\
    &= -a \frac{-\gamma(1 + \gamma A + \pi_0 e^{-\nu t}B y_1)}{(\gamma A - \gamma y_2)^2(1 - \gamma A y_2 + \pi_0 e^{-\nu t}B y_1)^2} \left(-\gamma(1-A) - \pi_0 e^{-\nu t}B y_1 + 2\gamma^2 A y_2 \right) \\
    & -b \frac{-\gamma(1 + \gamma A -\pi_1 e^{-\nu t}B y_1)}{(\gamma A - \gamma y_2)^2(1 - \gamma A y_2 -\pi_1 e^{-\nu t}B y_1)^2} \left(-\gamma(1-A) +\pi_1 e^{-\nu t}B y_1 + 2\gamma^2 A y_2 \right) \\
    \frac{\partial^3}{\partial y_2^3} \ln H &= aF_{\pi_0,\nu t}(y_1) \left(-\frac{2}{s_1^3} \frac{\partial s_1}{\partial y_2} + \frac{1}{s_1^2}\frac{\partial^2 s_1}{\partial y_2^2}\right) + bF_{\pi_1,\nu t}(y_1) \left(-\frac{2}{s_2^3} \frac{\partial s_2}{\partial y_2} + \frac{1}{s_2^2}\frac{\partial^2 s_2}{\partial y_2^2}\right) \\
    &= \frac{2a\gamma(1+\gamma A + \pi_0 e^{-\nu t}B y_1)}{(\gamma A - \gamma y_2)^2(1 - \gamma A y_2 + \pi_0 e^{-\nu t}B y_1)^3}(-\gamma(1-A) - \gamma \pi_0 e^{-\nu t}By_1 + 2\gamma^2 A y_2) \\
    &- \frac{a \gamma(1 + \gamma A + \pi_0 e^{-\nu t}B y_1)}{(\gamma A - \gamma y_2)^2(1 - \gamma A y_2 + \pi_0 e^{-\nu t}B y_1)^2}2 \gamma^2 A \\
    &+ \frac{2b\gamma(1+\gamma A - \pi_1 e^{-\nu t}B y_1)}{(\gamma A - \gamma y_2)^2(1 - \gamma A y_2 - \pi_1 e^{-\nu t}B y_1)^3}(-\gamma(1-A) + \gamma \pi_1 e^{-\nu t}By_1 + 2\gamma^2 A y_2) \\
    &- \frac{b \gamma(1 + \gamma A - \pi_1 e^{-\nu t}B y_1)}{(\gamma A - \gamma y_2)^2(1 - \gamma A y_2 - \pi_1 e^{-\nu t}B y_1)^2}2 \gamma^2 A \\
    &= -2a\frac{(1+\gamma A + \pi_0 e^{-\nu t}B y_1)}{( A - y_2)^2(1 - \gamma A y_2 + \pi_0 e^{-\nu t}By_1)^2} \\
    &\times \left(1 - A + \pi_0 e^{-\nu t}B y_1 - 2 \gamma y_2 + \gamma^2 A + \gamma^3 A + \gamma^2 A \pi_0 e^{-\nu t}B y_1 \right) \\
    &-2b\frac{(1+\gamma A - \pi_1 e^{-\nu t}B y_1)}{(A - y_2)^2(1 - \gamma A y_2 - \pi_1 e^{-\nu t}By_1)^2} \\
    &\times \left(1 - A - \pi_1 e^{-\nu t}B y_1 - 2 \gamma y_2 + \gamma^2 A + \gamma^3 A - \gamma^2 A \pi_1 e^{-\nu t}B y_1 \right).
\end{align*} Evaluating at $y_1 = 0$ and $y_2 = 1$ gives first partial derivatives equal to \begin{align*}
    \frac{\partial}{\partial y_1} \ln H &\rightarrow a\pi_0 e^{-\nu t}e^{-\mu t} - b \pi_1 e^{-\nu t}e^{- \mu t} \\
    \frac{\partial}{\partial y_2} \ln H &\rightarrow a e^{(\lambda - \mu)t} + b e^{(\lambda - \mu)t} = M e^{(\lambda - \mu)t}.
\end{align*} The second derivatives evaluate to \begin{align*}
    \frac{\partial^2}{\partial y_1^2} \ln H &\rightarrow -a \pi_0^2 e^{-2\nu t}e^{-2\mu t} - b \pi_1^2 e^{-2\nu t}e^{-2\mu t} \\
    \frac{\partial^2}{\partial y_1 y_2} \ln H & \rightarrow -a\pi_0 e^{-\nu t} e^{-2\mu t}\frac{\mu e^{\lambda t} - \lambda e^{\mu t}}{\mu - \lambda} + b \pi_1 e^{-\nu t} e^{-2\mu t}\frac{\mu e^{\lambda t} - \lambda e^{\mu t}}{\mu - \lambda} \\
    \frac{\partial^2}{\partial y_2^2} \ln H &\rightarrow -a e^{(\lambda - 3\mu)t}(\mu - \lambda)^{-2}\frac{\partial s_1}{\partial y_2}(0,1) \\
    &-b e^{(\lambda - 3\mu)t}(\mu - \lambda)^{-2}\frac{\partial s_1}{\partial y_2}(0,1) \\
    &= -M e^{(\lambda - 3\mu)t}(\mu - \lambda)^{-2}\frac{\partial s_1}{\partial y_2}(0,1).
\end{align*} The third derivatives evaluate to \begin{align*}
    \frac{\partial^3}{\partial y_1^3}\ln H &\rightarrow 2a \pi_0^3 e^{-3\nu t}e^{-3\mu t} - 2b \pi_1^3 e^{-3\nu t}e^{-3\mu t} \\
    \frac{\partial^3}{\partial y_1^2 y_2} \ln H & \rightarrow 2a \pi_0^2 e^{-2\nu t} e^{-3\mu t}\frac{\mu e^{\lambda t} - \lambda e^{\mu t}}{\mu - \lambda} + 2b \pi_1^2 e^{-2\nu t} e^{-3\mu t} \frac{\mu e^{\lambda t} - \lambda e^{\mu t}}{\mu - \lambda} \\
    \frac{\partial^3}{\partial y_1 y_2^2} \ln H&\rightarrow 2a \pi_0 e^{-\nu t} e^{-3\mu t} \left(\frac{\mu e^{\lambda t} - \lambda e^{\mu t}}{\mu - \lambda}\right)^2 - 2b \pi_1 e^{-\nu t} e^{-3\mu t}\left(\frac{\mu e^{\lambda t} - \lambda e^{\mu t}}{\mu - \lambda}\right)^2 \\
    \frac{\partial^3}{\partial y_2^3} \ln H & \rightarrow M e^{(\lambda - 3\mu)t}(\mu - \lambda)^{-2}\left(-2 \frac{1}{(\mu - \lambda)^2 e^{2\mu t}} \frac{\partial s_1}{\partial y_2}(0,1) + \frac{\partial^2 s_1}{\partial y_2^2}(0,1)\right)
\end{align*} Note that pure partial derivatives in $y_2$ are determined by the numbers $M, \lambda t, \mu t$.

\subsection{Probability of the first digit}

Recall that $p_{\vec{x}}^{\sigma}(t)$ is the probability  of observing  $\sigma\in\{0,1\}$ as the \textit{first digit} of a sequence at time $t$, under the TKF91 edge process starting at $\vec{x}$.  For the case $\lambda_0=0$ (i.e. there is no immortal link), this probability is given in Lemma \ref{L:theta_v} below. 

To simplify expressions we let $\gamma=\lambda/\mu$ and $\beta_t=e^{(\lambda-\mu)t}$. The probability that a nucleotide dies and has no descendant at time $t$ is
\begin{equation}\label{Def:eta}
    \eta(t) =
\frac{1- \beta_t}{1-\gamma \beta_t}= \P(|\mathcal{I}_t|=0\;\big|\;|\mathcal{I}_0|=1).
\end{equation}
We also let
\begin{equation}\label{Def:phipsi}
\phi(t)=\frac{
-\psi(t)+ 1 - \eta(t)}{1-\eta(t)} \quad \text{and}\quad
\psi(t)= e^{-(\mu+\nu)t}.
\end{equation}

The following is similar to  \cite[Lemma 1]{fan2020statistically}, but here there is no immortal link.
\begin{lemma}[Probability of first digit]\label{L:theta_v}
For any digit $\sigma \in\{0,1\}$, any initial sequence $\vec{x}=(x_i)_{i=1}^{|\vec{x}|}$ and time $t\in(0,\infty)$,
 	\begin{equation}\label{E:Q^sigma}
		p^{\sigma}_{\vec{x}}(t)= \pi_{\sigma}\,\phi(t)\,\big[1-\big(\eta(t)\big)^{|\vec{x}|}\big]\,+\,\psi(t)\sum_{i=1}^{|\vec{x}|}1_{\{x_i=\sigma\}}\big(\eta(t)\big)^{i-1}. 
	\end{equation}
\end{lemma}

\begin{proof}
Let $\{\mathcal{I}_t(1) = \sigma\}$ be the event that the \textit{first digit} of a sequence at time $t$ is $\sigma$.  We decompose this event according to the ancestor of this first digit.  To simplify the	notation we set $\eta:=\eta(t)$.

For   $1\leq i\leq |\vec{x}|$, we let	$\mathcal{K}_i$ be the event that the first nucleotide is the descendant of $x_i$. Note that $\mathcal{K}_i$ is the event that the normal link $x_i$ either survives or dies but has at least $1$ descendant, while the offspring of all previous links die. Then 
$
 \P^{\vec{x}}(\mathcal{K}_i)=\eta^{i-1}.
$
and  by the law of total probability,
	\begin{align}
	p^{\sigma}_{\vec{x}}(t)=&\,\P^{\vec{x}}\left(\mathcal{I}_t(1) 
	= \sigma\right) =\, \sum_{i=1}^{|\vec{x}|} \P^{\vec{x}}\big(\mathcal{I}_t(1) = \sigma,\;\mathcal{K}_i\big). \label{theta_v3}
	\end{align}

We use some notation of~\cite{Thorne1991}, where a nucleotide is also referred to as a ``normal link'' and a generic such nucleotide is denoted $``\star"$. We define
	\begin{align*}
	p_k:=p_k(t)&:=\P_{\star}( \text{normal link }``\star" \text{ survives and has }k\text{ descendants including itself}),\\
	 p^{(1)}_k :=p^{(1)}_k(t)&:=\P_{\star}( \text{normal link }``\star" \text{ dies and has }k\text{ descendants}).
	\end{align*}
	These probabilities are explicitly found in  \cite[Eq.~(8)-(10)]{Thorne1991} by solving the differential equations governing the underlying birth and death processes:
	\begin{align*}
	&\Bigg\{\begin{array}{ll}
	p_n(t) &= e^{-\mu t}(1-\gamma\eta(t))[\gamma\eta(t)]^{n-1}\\
	p^{(1)}_n(t) &= (1-e^{-\mu t}-\eta(t))(1-\gamma\eta(t))[\gamma\eta(t)]^{n-1}
	\end{array}\qquad \text{for }n\ge 1, 
	\\
	&\Bigg\{\begin{array}{ll}
	p_0(t) &=0\\
	p^{(1)}_0(t) &= \eta(t).
	\end{array}
	\end{align*}

	We now compute each term on the RHS of \eqref{theta_v3}. For $1\leq i\leq |\vec{x}|$, we let $S_i$ be the event that $x_i$ survives, which has probability $e^{-\mu t}$. Then
    \begin{equation*}
        \P^{\vec{x}}(\mathcal{K}_i \cap S_i)=(p^{(1)}_0)^{i-1}e^{-\mu t}=\eta^{i-1}\,e^{-\mu t},
    \end{equation*}
    since  
    $(p^{(1)}_0)^{i-1}$ is the probability that all $\{v_j\}_{j=1}^{i-1}$ were deleted and left no descendant
    and $S_i$ is independent of the previous events.
	Moreover, we have $\P^{\vec{x}}(\mathcal{I}_t(1) = \sigma\,|\,\mathcal{K}_i \cap S_i)= f_{x_i \sigma}$, where
	\begin{equation*}
	f_{ij}:=f_{ij}(t) 
	=\pi_{j}(1-e^{-\nu t})+e^{-\nu t}1_{\{i=j\}}
	\end{equation*}
	is the transition probability that a nucleotide is of type $j$ after time $t$, given that it is of type $i$ initially. 
    Letting $S_i^c$ be the complement of $S_i$, then
     \begin{equation*}
      \P^{\vec{x}}(\mathcal{K}_i \cap S_i^c)=\eta^{i-1}\, \left(\sum_{k\geq 1}p^{(1)}_k \right) =\eta^{i-1}\, (1-e^{-\mu t}-\eta),
    \end{equation*}
    and $\P^{\vec{x}}(\mathcal{I}_t(1) = \sigma\,|\,\mathcal{K}_i \cap S^c_i)=\pi_{\sigma}$. Therefore,
    \begin{align}
    &\P^{\vec{x}}(\mathcal{I}_t(1) = \sigma\,,\,\mathcal{K}_i) \notag\\
    &=\P^{\vec{x}}(\mathcal{I}_t(1) = \sigma\,|\,\mathcal{K}_i \cap S_i)\,\P^{\vec{x}}(\mathcal{K}_i \cap S_i)+\P^{\vec{x}}(\mathcal{I}_t(1) = \sigma\,|\,\mathcal{K}_i \cap S^c_i)\,\P^{\vec{x}}(\mathcal{K}_i \cap S_i^c)   \notag\\
    &= \eta^{i-1}\,\left( f_{x_i \sigma}\,e^{-\mu t} + \pi_{\sigma} (1-e^{-\mu t}-\eta) \right). \label{theta_v3_2}
    \end{align}

	Putting \eqref{theta_v3_2} into \eqref{theta_v3}, 
	\begin{align*}
    \P^{\vec{x}}\left(\mathcal{I}_t(1) = \sigma\right)
    &= 
    \sum_{i=1}^{|\vec{x}|}
    \left[
    \eta^{i-1}\,\left( f_{x_i \sigma}\,e^{-\mu t} + \pi_{\sigma} (1-e^{-\mu t}-\eta) \right)
    \right]\\
	&=  e^{-\mu t} \left[\sum_{i=1}^{|\vec{x}|} \eta^{i-1}f_{x_i \sigma}\right]+ (1-e^{-\mu t} -\eta)\pi_{\sigma} \frac{1-\eta^{|\vec{x}|}}{1-\eta},
	\end{align*}
	which is exactly \eqref{E:Q^sigma} upon further re-writing
	\begin{equation*}
	\sum_{i=1}^{|\vec{x}|}\eta^{i-1}f_{x_i \sigma}
	=\pi_{\sigma}(1-e^{-\nu t})\frac{1-\eta^{|\vec{x}|}}{1-\eta}+e^{-\nu t}\sum_{i=1}^{|\vec{x}|}1_{\{x_i=\sigma\}}\eta^{i-1}.
	\end{equation*}
	The proof of Lemma \ref{L:theta_v} is complete.
\end{proof}

\subsection{About the special case $\lambda=\mu$}

For completeness, here we consider the special case $\lambda = \mu$ (i.e. $\gamma = 1$). 

\begin{prop}\label{L:InvertLengthSpecial}
Let $(\mathcal{I}_t)_{t\geq 0}$ denote the indel process  in Definition \ref{Def_BiINDEL} and suppose it is known that $\lambda = \mu$. Given the distribution of the length $|\mathcal{I}_t|$ at time $t$ starting at an unknown sequence $\vec{x}$, we can identify the ancestral sequence length $M=|\vec{x}|$ and the parameter $\mu t$ as follows: 
$M=\mathbb{E}_{M}[|\mathcal{I}_t|]$, and $$\mu t=\ln\left(\frac{1}{M^2 + 4M+1} \left(\mathbb{E}_{M}[|\mathcal{I}_t|^2] + 4M + 1 \right ) \right).$$
\end{prop}

\begin{proof}
    In the case where $\lambda = \mu$, we need only the description of $\mathbb{E}_{M}[|\mathcal{I}_t|^2]$ given in \cite{THATTE200658}.  With no immortal link contribution, it has the property that $$\frac{d}{dt}\mathbb{E}_{M}[|\mathcal{I}_t|^2] = (3\lambda - 2\mu)\mathbb{E}_{M}[|\mathcal{I}_t|^2] + (3\lambda + \mu)\mathbb{E}_{M}[|\mathcal{I}_t|] + \lambda$$ with initial condition $\mathbb{E}_{M}[|\mathcal{I}_0|^2] = M^2$. When $\lambda = \mu$, this simplifies to $$\frac{d}{dt}\mathbb{E}_{M}[|\mathcal{I}_t|^2] = \mu \mathbb{E}_{M}[|\mathcal{I}_t|^2] + 4\mu M + \mu.$$ This is a separable differential equation with solution $$\mathbb{E}_{M}[|\mathcal{I}_t|^2] = (M^2 + 4M+1)e^{\mu t} - 4M - 1.$$ As $M$ is identifiable from $\mathbb{E}_{M}[|\mathcal{I}_t|]$, we find that $\mu t$ is identifiable from the second moment.  
\end{proof}

\bibliographystyle{plain}
\bibliography{Invert}

\end{document}